\documentclass[reqno]{amsart}
\usepackage{amsmath}
\usepackage{amssymb}
\usepackage{amsthm}
\usepackage{leftidx}
\usepackage{tikz-cd}
\usepackage{appendix}
\usepackage{url}

\newtheorem{lem}{Lemma}[section]
\newtheorem{thm}{Theorem}[section]
\newtheorem{prop}{Proposition}[section]
\newtheorem{cor}{Corollary}[section]
\newtheorem{defx}{Definition}[section]

\newcommand{\w}{\omega}
\newcommand{\J}{\mathcal{J}}

\newcommand{\C}{\mathbb{C}}
\newcommand{\wre}{\omega_{\text{red}}}

\newcommand{\nr}{n_R}

\newcommand{\oo}{\mathcal{O}}
\newcommand{\tth}{\theta}
\newcommand{\ber}{\text{Ber }}

\newcommand{\tred}{\text{red}}
\newcommand{\Dbar}{\bar{D}}
\newcommand{\msp}{\mathfrak{M}_{g;\nr}}

\newcommand{\stsh}{\mathcal{O}}

\newcommand{\F}{\mathcal{F}}
\newcommand{\Ll}{\mathcal{L}}
\newcommand{\Kk}{\mathcal{K}}
\newcommand{\ee}{\mathcal{E}}
\newcommand{\bz}{\bar{z}}
\newcommand{\spec}{\text{Spec }}
\newcommand{\bpk}{[\partial_{z_k} \, | \, \partial_{\tth_k}]}

\title[On the Super Mumford Form in the Presence of Punctures]{On the Super Mumford Form in the presence of Ramond and Neveu-Schwarz Punctures}
\author[Daniel J. Diroff]{Daniel J. Diroff \\
University of Minnesota \\
Email: dirof003@umn.edu}

\begin{document}

\begin{abstract} We generalize the result of \cite{vormum} to give an expression for the super Mumford form $\mu$ on the moduli spaces of super Riemann surfaces with Ramond and Neveu-Schwarz punctures. In the Ramond case we take the number of punctures to be large compared to the genus. We consider for the case of Neveu-Schwarz punctures the super Mumford form over the component of the moduli space corresponding to an odd spin structure. The super Mumford form $\mu$ can be used to create a measure whose integral computes scattering amplitudes of superstring theory. We express $\mu$ in terms of local bases of $H^0(X, \w^j)$ for $\w$ the Berezinian line bundle of a family of super Riemann surfaces.
\end{abstract}

\maketitle

\section*{Introduction}
Due to relatively recent computations done by E. D'Hoker and D. H. Phong \cite{DPhong} and new ideas pushed forward by E. Witten \cite{wit4}, the role of supergeometry in superstring perturbation theory has been revived from what it once was in the 1980s. However, the task of computing superstring scattering amplitudes have proved difficult due to many complications boiling down to the fact that the underlying supergeometry was not completely understood. 

Scattering amplitudes in superstring theory are expressed as Berezin integrals over various moduli spaces of super Riemann surfaces. One might hope that such integrals would be computable via expressing supermoduli space as a fiber bundle over a bosonic reduced space, allowing one to integrate in the odd directions fiberwise. In fact, this is exactly the technique utilized in the D'Hoker and Phong results. However, this assumption was only valid for low genus, as it was shown in a recent paper by R. Donagi and E. Witten \cite{donwit} that in general supermoduli space \emph{is not} a fiber bundle over its reduced space. This notion is significant in supergeometry and is known as \emph{splitness}.

Essentially, one says a supermanifold is split if it can be expressed as such a fiber bundle over a bosonic base. It is known that every $C^{\infty}$ supermanifold is indeed split \cite{man1}. Thus in principal the theory of smooth supermanifolds is contained in the theory of exterior algebra vector bundles over a smooth manifold. However, holomorphic methods have proved to be very useful in studying super Riemann surfaces and their moduli as holomorphic or complex supermanifolds need not be split. Thus holomorphic supergeometry is central in understanding computations of superstring scattering amplitudes.

In bosonic string theory, the $g$ loop contribution to the partition function can be written as the integral
$$ Z_g = \int_{\mathcal{M}_g} d\pi_g, $$
where $\mathcal{M}_g$ is the usual moduli stack of Riemann surfaces of genus $g$ and $d\pi_g$ is the so-called Polyakov measure. Suppose we have a universal family $\mathcal{C}_g$ over $\mathcal{M}_g$ and let $\pi: \mathcal{C}_g \to \mathcal{M}_g$ denote the projection. 

In a famous theorem due to Belavin and Knizhnik, the Polyakov measure was shown to be the modulus squared of a trivializing section of a holomorphic line bundle on $\mathcal{M}_g$,
$$ d\pi_g = \mu_g \wedge \overline{\mu}_g. $$
The form $\mu_g$ is called a Mumford form and it is a section exhibiting the Mumford isomorphism
$$ \left( \text{det }  \pi_*\Omega \right)^{13} \otimes \left( \text{det }  R^1\pi_*\Omega \right)^{-13} \cong \det \pi_*\Omega^2 \otimes \left( \det R^1\pi_*\Omega^2 \right)^{-1},$$
where $\Omega$ is the sheaf of relative differentials on $\mathcal{C}_g$. Here and henceforth, powers of vector bundles, sheaves and vector spaces stand for tensor powers.

In the super case, the object one integrates over in computations of superstring scattering amplitudes is slightly more complicated than simply $\mathfrak{M}_g$, see \cite{wit1}. Nevertheless there still is a relevant canonical super Mumford isomorphism,
$$ (\ber \pi_*\w )^5 \otimes (\ber R^1\pi_*\w )^{-5} \cong \ber \pi_*\w^3 \otimes \left( \ber R^1\pi_*\w^3 \right )^{-1}  $$
for $\w$ the relative Berezinian sheaf of a family of super Riemann surfaces of genus $g$. The trivializing section that exhibits the above isomorphism is called the super Mumford form. Such a form is useful in the super case in very much the same way as that of the bosonic Mumford form, as sections of $\ber \pi_*\w^3$ are super volume forms on $\mathfrak{M}_g$. In a paper by A. Voronov \cite{vormum}, an explicit formula of the super Mumford form was computed over the odd-spin component of $\mathfrak{M}_g$. 

In this paper we expand on those ideas and produce explicit formulas for the analogous super Mumford forms over the moduli spaces $\mathfrak{M}_{g;n_R}$ and $\mathfrak{M}_{g;n_{NS}}$ of genus $g \geq 2$ super Riemann surfaces with Ramond or Neveu-Schwarz punctures. In both cases we work under some assumptions regarding the local freeness of the sheaves $R^i\pi_*\w^j$. The specifics are given at the end of Section \ref{sec1}. In the Ramond case we furthermore impose the condition that the number of Ramond punctures $\nr$ be strictly greater than $6g-6$. 

We then discuss how these formulae give rise to a physically relevant measure. By explicit formulas, we mean those written in terms of chosen sections of natural sheaves defined on the moduli spaces.

The paper is divided into 3 sections. The main results (Theorems \ref{mainthm1} and Corollary \ref{maincor}) are found in Sections 2 and 3 where the explicit formulas of the relevant super Mumford forms are presented. Section 1 is devoted to setting up notation and briefly reviewing the basic notions needed for the rest of the paper.  Appendices appear after Section 3 containing a few technical lemmas used in the main arguments.

\section{Super Riemann Surfaces and Other Preliminaries} \label{sec1}

\subsection{Definitions and Basic Notions}\mbox{}

We briefly review some basic definitions and notions to setup notation. \emph{Super Riemann surfaces} are a certain class of complex supermanifolds of dimension $1 | 1$, which carry an additional piece of structure. These play the role of superstring worldsheets and their theory very closely parallels that of classical Riemann surfaces.

We are interested in the moduli of these objects and thus have the following notion of a family.

\begin{defx}
A family of super Riemann surfaces is a family of complex supermanifolds $\pi: X \to S$ of relative dimension $1 | 1$ equipped with a maximally non-integrable distribution $\mathcal{D}$ of rank $0 | 1$, i.e. an odd subbundle of the relative tangent bundle $\mathcal{T}_{X/S}$ such that the Lie bracket induces the isomorphism
$$ [\cdot,\cdot] : \mathcal{D}^2 \xrightarrow{\sim} \mathcal{T}_{X/S}/\mathcal{D}.
$$
\end{defx}
Locally one can always find relative coordinates $x | \theta$ such that the distribution $\mathcal{D}$ is generated by the odd vector field $D_{\theta} = \partial_{\theta} + \theta \partial_x$, such coordinates are called \emph{superconformal}. We say a change of coordinates $y | \zeta$ is \emph{superconformal} if $\mathcal{D}_{\zeta}$ and $\mathcal{D}_{\theta}$ are $\stsh_X$-multiples of each other. Throughout this paper we will sometimes refer to a family of super Riemann surfaces as a \emph{family of SUSY curves} or simply by a \emph{SUSY family}.

It is well known \cite{modSRS} that if the base $S$ is reduced, we essentially get a classical object, namely a family of spin curves.

\begin{prop} \label{jspin}
Let $\pi: X \to S$ be a family of super Riemann surfaces over a reduced base $S$. Let $\J \subset \stsh_X$ denote the sheaf of ideals generated by all odd elements. Then
\begin{enumerate}
    \item $\mathcal{J}$ is a locally free $\stsh_{\tred}$ module of rank $0 | 1$,
    \item $\J^* = \mathcal{H}om_{\stsh_{\tred}}(\J, \stsh_{\tred}) \cong \mathcal{D}_{\tred}$
    \item $\Pi \mathcal{J}$ becomes a relative spin structure on the family $X_{\tred} \to S$, i.e.
$$ (\Pi \mathcal{J}^{\otimes 2}) = \J^{\otimes 2} \cong \Omega_{X_{\tred}/S}^1, $$
\end{enumerate}
where $\Pi$ is the parity reversing functor.
\end{prop}
In other words, we have that the reduction of supermoduli space $\mathfrak{M}_g$ of super Riemann surfaces of genus $g$ is the moduli space $\mathcal{SM}_g$ of genus $g$ Riemann surfaces equipped with a spin structure.

We will also use the notion of a \emph{family of supercurves}, by which we mean simply a family $\pi: X \to S$ of complex supermanifolds of relative dimension $1|1$. Then a family of SUSY curves is a family of supercurves with the extra data of the odd distribution $\mathcal{D}$.

If $\mathcal{F}$ denotes a rank $m \,|\, n$ locally free sheaf on $X$ with local generators \newline
$e_1, \dots, e_m \, |\, \zeta_1, \dots, \zeta_n$ we denote by
$$
[e_1, \dots, e_m \, |\, \zeta_1, \dots, \zeta_n]
$$
a local generator of the invertible sheaf $\ber \mathcal{F}$.

Given a family of supercurves $\pi:X \to S$, if a dualizing sheaf exists we will denote it by $\w$. In the case $\pi:X \to S$ is a family of SUSY curves it was shown in \cite{elmsg} that a dualizing sheaf exists and is the Berezinian of the relative cotangent sheaf $\w = \ber X/S := \ber \Omega^1_{X/S}$. In fact, from the short exact sequence 
$$
0 \longrightarrow \mathcal{D} \longrightarrow \mathcal{T}_{X/S} \longrightarrow \mathcal{D}^2 \longrightarrow 0
$$
one concludes $\w \cong \mathcal{D}^{-1}$.

On any complex supermanifold $X$ one can construct the sheaf of \emph{smooth superfunctions}, denoted by $\ee$. Loosely, $\mathcal{E}$ can be defined by the condition that it is locally generated by functions $x_k, \bar{x}_k\, |\, \theta_k, \bar{\theta}_k$ if $x_k \, | \, \theta_k$ are local coordinates for X. This has been made precise in a short paper of Haske and Wells \cite{smoothfromcomplex}. We will think of $z_k, \bar{z}_k, \zeta_k, \bar{\zeta}_k$ as generators of $\ee$ analogous as to what is common in complex analysis, so that locally every smooth superfunction $f$ is of the form (abbreviating the indices and using the usual multi-index notation)
$$
f(z, \bar{z} | \zeta, \bar{\zeta}) =\sum_{I,J} f_{IJ}(z, \bz)\zeta_I \bar{\zeta}_J
$$
for ordinary ($\C$-valued) smooth functions $f_{IJ}$. The sheaf $\ee$ then naturally has complex conjugation. 

By a smooth section of a complex super vector bundle $\mathcal{F}$, we mean a section of the sheaf $\mathcal{F}_{\ee} := \mathcal{F} \otimes_{\stsh} \mathcal{E}$. Furthermore we denote by $\overline{\mathcal{F}}$, the complex conjugate vector bundle of $\mathcal{F}$. Of particular interest is the \emph{smooth} Berezinian sheaf $\omega \otimes \overline{\w} \otimes \ee =: |\omega|^2$, as its sections yield natural objects that can be integrated over the \emph{entire} complex supermanifold $X$.

\bigskip

\subsection{$\bar{D}$-Cohomology}~ \label{dbarcoho}

Here we follow very closely \cite{RSV} and \cite{gid2}. Suppose we have a family of supercurves $\pi:X \to S$. For each $p,q \in \mathbb{Z}$ one can consider the sheaves
$$ 
\omega^{p,q} = \omega^{\otimes p} \otimes \overline{\omega}^{\otimes q} \otimes \ee. 
$$
Then $\omega_{\ee} = \omega^{1,0}, |\w|^2 = \w^{1,1}$ and we have well-defined operators $D$ and $\Dbar$
$$ 
D: \ee \to \omega^{1,0}, \hspace{.15in} \Dbar: \ee \to \omega^{0,1}
$$
given in local superconformal coordinates $x \, |\, \theta$ by
$$ \Dbar(f) = \Dbar_{\bar{\tth}} f(x, \bar{x}|\tth, \bar{\tth}) \, [d\bar{x} \, | \, d\bar{\tth} ]= \left( \frac{\partial f}{\partial \bar{\tth}} + \bar{\tth}\frac{\partial f}{\partial \bar{x}}  \right) \, [d\bar{x} \, | \, d\bar{\tth} ]. $$
and similarly for $D$. It is easy to see that this definition does not depend on coordinates by assumption on how $D$ and $\bar{D}$ transform. The nice observation is in in the following proposition that was discussed and discovered in \cite{RSV} and \cite{gid2}.

\begin{prop} \label{dbarSES}
For any locally free $\stsh_X$-module $\mathcal{F}$, after extending $\Dbar$ we have the exact sequence
$$
 0 \longrightarrow \mathcal{F} \longrightarrow \mathcal{F} \otimes \ee \overset{\Dbar}{\longrightarrow} \mathcal{F} \otimes \omega^{0,1} \longrightarrow 0.
$$
\end{prop}

This gives us an interpretation of the cohomology groups of a super vector bundle $\mathcal{F}$, which will prove useful for the computations done below.

\begin{cor}
For any super vector bundle $\mathcal{F}$ on $X$, the cohomology of the sheaf $\mathcal{F}$ is computed via the exact sequence given in the prior proposition.
\end{cor}
\begin{proof}
Immediate from the proposition and the fact that both $\mathcal{F} \otimes \ee$ and $\mathcal{F} \otimes \omega^{0,1}$ are acyclic.
\end{proof}

\bigskip

\subsection{Connection Between the Berezinian and One-Forms}~

In \cite{RSV} an interesting and useful connection was made between one-forms and sections of the Berezinian on a super Riemann surface. By definition of a SUSY family $\pi:X \to S$ one sees that we have the short exact sequence
$$ 0 \longrightarrow \mathcal{D}^{-2} \longrightarrow \Omega^1_{X/S} \longrightarrow \mathcal{D}^{-1} \longrightarrow 0, 
$$
(Here we write $\mathcal{D}^{-2}$ for $(\mathcal{D}^{-1})^{\otimes 2}$). Combining $\Omega^1_{X/S} \to \mathcal{D}^{-1}$ and the isomorphism $\mathcal{D}^{-1} \cong \w$, we get a natural map taking holomorphic one-forms to sections of the Berezinian. In local coordinates $z | \tth$ this is
$$
f(z|\tth) dz + g(z|\tth) d\tth \mapsto ( g(z|\tth) + f(z|\tth) \tth) [dz \, | \, d\tth].
$$

This map cannot be an isomorphism as $\Omega^1_X$ has rank $1|1$ while $\w$ is of rank $0|1$, however in \cite{RSV} it was noticed that upon restriction to $d$-\emph{closed} one-forms, we do get an isomorphism (here $d$ is the usual exterior derivative). The inverse map we denote by $\alpha: \w \to Z^1_X := \{ \text{closed holomorphic one-forms} \}$. It is given in coordinates as, for $\sigma = f(z|\tth) [dz \, | \, d\tth]$,
\begin{equation} \label{alpha}
\alpha(\sigma) : = d\tth f(z|\tth) + \varpi D_{\tth}f(z|\tth),
\end{equation}
where $\varpi := dz - \tth d\tth$ is the local generator of $\mathcal{D}^{-2}$ and $D_{\tth}$ is the usual local generator of the distribution. Note that above we have followed the convention in \cite{RSV} and have written the coefficient functions \emph{to the right} of the forms $d\tth$ and $\varpi$. 

One can check that the local coordinate definition (\ref{alpha}) is well-defined and gives a genuine map $\alpha: \w \to Z^1_X$. A coordinate invariant description of $\alpha$ is described in \cite{wit1}; it is related to the notion of \emph{picture number} and \emph{picture changing operators} in string theory. We will not need these notions here, thus the definition (\ref{alpha}) suffices. 

Hence, we can associate to each section of the Berezinian $\sigma$ a closed holomorphic one-form $\alpha(\sigma)$.

\bigskip

\subsection{Residues and Contour Integrals}~

For what follows we will use the notion of an integral of a $k$-form over a $k | 0$ dimensional submanifold. The definition is the following: given a submanifold $N = N^{k | 0}$ of a $m|n$ dimensional supermanifold $M = M^{m|n}$ described by a map $\varphi: N \to M$, for $\eta$ a $k$-form on $M$, $\varphi^*\eta$ is a $k$-form (a top form) on $N$ and hence can be integrated, we write
$$
\int_N \eta : = \int_{N}\varphi^*\eta.
$$

On a super Riemann surface $X$, this allows us to discuss residues of meromorphic sections of the Berezinian $\w$ at points in the reduced space. Namely, for $\sigma$ a meromorphic section of $\w$ and $p$ a point of $X_{\tred}$, we define the residue of $\sigma$ at $p$ by the formula
\begin{equation} \label{resdef}
\text{res}_{p} \sigma := \frac{1}{2 \pi i} \oint_{\gamma} \alpha(\sigma).
\end{equation}
Here $\gamma: I \to X $ is any simple closed curve enclosing $p$, given by a map from an interval $I$ to $X$ and viewed as a real submanifold of dimension $1|0$. The map $\alpha$ is that of (\ref{alpha}). Stokes' Theorem gives then that the fact that $\alpha$ has image in closed one-forms guarantees that (\ref{resdef}) does not depend on the choice of curve $\gamma$.

In local coordinates $z | \tth $ such that $p$ is defined by $z = z_0$, $\tth = \tth_0$ we expand $\sigma$ in a Laurent series
$$
\sigma = \left ( \sum_{k = -N}^{\infty} (a_k + (\tth-\tth_0)b_k) (z-z_0 -\tth\tth_0)^k \right ) [dz \, | \, d\tth],
$$
then a straight forward calculation will show that
\begin{equation} \label{residue}
 \text{res}_p \sigma = b_{-1}.   
\end{equation}

Of course, the residue is coordinate independent in view of (\ref{resdef}), but (\ref{residue}) allows for straightforward computations. Additionally, we have the slightly more general formula, which is verified by direct computation; if one writes in the coordinates above $\sigma = (z-z_0-\tth\tth_0)^{-1} f(z|\tth) [dz \, | \, d\tth]$ then
\begin{equation} \label{res2}
\text{res}_p \sigma = (D_{\tth}f)(z_0 | \tth_0).
\end{equation}

\bigskip

\subsection{Serre Duality}~

For a SUSY family $\pi: X \to S$ we remarked above that the relative Berezinian sheaf is Serre dualizing  $\omega = \ber X/S$. That is for any invertible sheaf $\Ll$ (rank $1  | 0$ or $0 | 1$) we have for $i=0,1$ a perfect pairing \cite{smoothfromcomplex}, \cite{elmsg},
$$
R^i\pi_*(\Ll) \otimes R^{1-i}\pi_*(\Ll^* \otimes \omega) \to R^1\pi_*(\omega)
$$
along with a trace map
$$
\text{tr}: R^1\pi_*(\omega) \to \stsh_S.
$$
Using $\Dbar$ cohomology the trace map tr is simply the Berezin integral along the fibers. We will not prove this, but one can see that indeed the Berezin integral gives a map
$$
\int_{X/S}: \omega^{1,1} \to \stsh_S
$$
and an analog of Stokes' theorem can be established so that the Berezin integral of $\Dbar$-exact fields are zero, which implies that $\int_{X/S}$ decends to a map on the quotient $R^1\pi_*(\omega) \rightarrow \stsh_S$.

\bigskip

\subsection{Punctures}~

Scattering amplitudes of superstring theory are written as integrals over moduli spaces of slightly more general objects than strictly super Riemann surfaces. These are \emph{punctured super Riemann surfaces} which we discuss now.

There are two types of punctures one can consider in the theory of SUSY curves, known as Neveu-Schwarz and Ramond punctures. Neveu-Schwarz punctures are more familiar, while Ramond punctures are a bit exotic. We focus on Ramond punctures first.

\medskip

\subsubsection{Ramond Punctures}~ \label{ramondpunctures}

Suppose $\pi: X \to S$ of $1 | 1$ is a family of supercurves along with an odd distribution $\mathcal{D} \subset \mathcal{T}_{X/S}$ such that the Lie bracket
$$
\mathcal{D}^{\otimes 2} \xrightarrow{[\cdot, \cdot]} \mathcal{T}_{X/S}/\mathcal{D}
$$
fails to be an isomorphism along a relative divisor $\mathcal{F}$, in the sense that instead $[\cdot, \cdot]$ induces an isomorphism
$$
\mathcal{D}^{\otimes 2} \xrightarrow{[\cdot, \cdot]} \mathcal{T}_{X/S}/\mathcal{D} \otimes \stsh_X(-\mathcal{F}).
$$
In this case the family $\pi: X \to S$ is called \emph{a family of super Riemann surfaces with Ramond punctures} or a \emph{family of SUSY curves with Ramond punctures}. The divisor $\mathcal{F}$ is called the \emph{Ramond divisor}. If we write $\mathcal{F}$ as a sum of minimal divisors (irreducible divisors)
$$
\mathcal{F} = \sum_{k=1}^{n_R} \mathcal{F}_k,
$$
then each $\mathcal{F}_k$ is called a \emph{Ramond puncture}. It can be easily shown that the number of Ramond punctures is in fact even. One can think of a Ramond puncture as a ``puncture" in the distribution $\mathcal{D}$ itself.

Locally near a Ramond puncture $\mathcal{F}_k$ we can find a coordinate chart $z | \zeta$ so that $\mathcal{F}_k$ is given by $z = 0$ and that $\mathcal{D}$ is locally generated by $D^*_{\zeta} = \partial_\zeta + z \zeta \partial z$ (say, in the complex topology). Such coordinates are also called \emph{superconformal}. The usual exact sequence now becomes
$$ 0 \longrightarrow \mathcal{D} \longrightarrow \mathcal{T}_X \longrightarrow \mathcal{D}^{2}(\mathcal{F}) \longrightarrow 0. $$
Dualizing and taking Berezinians we conclude $\omega = \ber X/S \cong \mathcal{D}^{-1}(-\mathcal{F})$. In fact, in this case $\omega$ remains a relative dualizing sheaf.

We denote the moduli space of super Riemann surfaces with $\nr$ Ramond punctures by $\msp$. One has that \cite{wit1} the tangent sheaf $\mathcal{T}_{\msp}$ is given by $R^1\pi_*\mathcal{W}$ where $\mathcal{W}$ is the sheaf of infinitesimal automorphisms, namely the sheaf of vector fields that preserve the distribution $\mathcal{D}$ (in the sense that $[\mathcal{W}, \mathcal{D}] \subset \mathcal{D}$). In local coordinates one can easily verify \cite{vids} that as sheaves of $\mathbb{C}$-vector spaces $\mathcal{W} \cong \mathcal{D}^2$. 

\medskip

\subsubsection{Neveu-Schwarz Punctures}~ \label{nspuctures_sec}

Suppose now that $\pi: X \to S$ is a SUSY family. A \emph{Neveu-Schwarz} (NS) puncture is simply a section $s: S \to X$ of the map $\pi$. Such a section is locally $ \spec A \to \spec A[z | \zeta] $ and hence equivalent to a map of supercommutative rings $A[z | \zeta] \to A$ which in turn is simply a choice of an even and odd element of $A$. Hence it is common to say that an NS puncture is given in local superconformal coordinates by $z = z_0, \zeta = \zeta_0$ for some choice of even and odd functions $z_0, \zeta_0$ on the base $S$. 

Given an NS puncture $s$ we have a natural associated divisor using the distribution $\mathcal{D}$. Namely, we use $s$ to pullback $\mathcal{D}$ and then take its total space $s^*\mathcal{D}^{\text{tot}}$.
$$
\begin{tikzcd}
(s^*\mathcal{D})^{\text{tot}} \arrow[r] \arrow[d]
& \mathcal{D}^{\text{tot}} \arrow[d, "p"] \\
S \arrow[r, "s"] & X
\end{tikzcd}
$$
This gives a subvariety $(s^*\mathcal{D})^{\text{tot}} \to X$, which is of relative dimension $0 | 1$ over $S$. We will denote this subvariety associated to $s$ by div$(s)$. Given such a family $\pi:X \to S$ with $n_{NS}$ NS punctures $s_1, \dots, s_{n_{NS}}$, we denote by $N = \sum_{j=1}^{n_{NS}} \text{div}(s_j)$ the \emph{Neveu-Schwarz divisor}.

We denote the moduli space of super Riemann surfaces with $n_{NS}$ Neveu-Schwarz punctures by $\mathfrak{M}_{g;n_{NS}}$. Similar to the Ramond puncture case, we have that the tangent sheaf $\mathcal{T}_{\mathfrak{M}_{g;n_{NS}}}$ is given by $R^1\pi_*\mathcal{W}$ for $\mathcal{W}$ the sheaf of infinitesimal automorphisms. The difference here is that $\mathcal{W}$ now is the sheaf of vector fields that preserves the distribution $[\mathcal{W},\mathcal{D}] \subset \mathcal{D}$ and must vanish along the Neveu-Schwarz divisor $N$. In local coordinates one can then see \cite{vids} that $\mathcal{W} \cong \mathcal{D}^2(-N)$ as sheaves of $\mathbb{C}$-vector spaces.

\bigskip

\subsection{The Super Mumford Isomorphism}~

We follow \cite{vormum} closely. Let $\pi: X \to S$ denote a family of $1 | 1$ supercurves. For any locally free sheaf $\mathcal{F}$ on $X$ we can consider the invertible sheaf $B(\mathcal{F})$ on $S$, called the \emph{Berezinian of cohomology of} $\mathcal{F}$. If each $R^i\pi_* \mathcal{F}$ is locally free, then $B(\mathcal{F})$ is given by
$$
B(\mathcal{F}) = \otimes_{i} \, (\ber R^i\pi_*\mathcal{F})^{(-1)^i}.
$$
Moreover, for every short exact sequence of locally free sheaves on $X$
$$
0 \longrightarrow \mathcal{F}' \longrightarrow \mathcal{F} \longrightarrow \mathcal{F}'' \longrightarrow 0,
$$
we get a canonical isomorphism
$$
B(\mathcal{F}') \otimes B(\mathcal{F}'') \cong B(\mathcal{F}).
$$
Hence, in particular any isomorphism $f: \mathcal{F} \to \mathcal{G}$ induces an isomorphism $B(f): B(\mathcal{F}) \to B(\mathcal{G})$. 

For $\w = \ber \Omega^1_{X/S}$ the relative Berezinian, we set for each $j$,
$$ \lambda_{j/2} = B(\w^{\otimes j}). $$
Serre duality gives the canonical identifications $\lambda_{j/2} \cong \lambda_{(1-j)/2}$. The super Mumford isomorphism(s) are the following canonical isomorphisms amongst the $\lambda_{j/2}$.

\begin{prop} \label{smumiso}
For any family of $1 | 1$ supercurves $\pi: X \to S$ we have canonical isomorphisms
$$
\lambda_{j/2} \cong \lambda_{1/2}^{(-1)^{j-1}(2j-1)}.
$$
In particular,
$$
\lambda_{3/2} \cong \lambda_{1/2}^5.
$$
\end{prop}

Proofs of the super Mumford isomorphisms can be found in \cite{vormum} and \cite{vids} . We will denote by $\mu$ the trivializing section of $\lambda_{3/2} \lambda_{1/2}^{-5}$. Such an object is called the \emph{super Mumford form}. 

In the following we will consider two separate situations:
\begin{enumerate}
    \item $\pi: X \to S$ is a family of super Riemann surfaces of genus $g \geq 2$ with $\nr$ Ramond punctures such that:
\begin{enumerate}
    \item The sheaves $R^i\pi_*\w^j$ are locally free for $i=0,1$, $j=-2, -1, 0, 1$.
    \item $\nr > 6g-6$.
\end{enumerate}
    \item $\pi: X \to S$ is a family of super Riemann surfaces of genus $g \geq 2$ with $n_{\text{NS}}$ Neveu-Schwarz punctures such that:
\begin{enumerate}
    \item The sheaves $R^i\pi_*\w^j$ are locally free for $i=0,1$, $j=0, 1, 2, 3$.
    \item $\pi_*\w$ has rank $g | 1$.
\end{enumerate}
\end{enumerate}
In each case we produce a concrete proof of the corresponding super Mumford isomorphism $\lambda_{3/2} = \lambda_{1/2}^5$ and use it to produce the main results.

We make heavy use of the higher direct image sheaves of the relative Berezinain $R^i\pi_*\w^j$ and emphasize that in both situations we work under the \emph{assumption} that these sheaves are locally free. We pause to discuss this assumption more in depth at the end of Section \ref{somerierochcalc}. 

The conditions listed for the Neveu-Schwarz case speak to the fact that we work over the component of the moduli space $\mathfrak{M}_{g;n_{NS}}$ corresponding to an odd spin structure. Hence, fiberwise $\Pi \w_{\tred}$ gives an odd nondegenerate theta characteristic.

We begin with the Ramond case.

\section{The Ramond Puncture Case}

Here we derive our first main result: an explicit formula for the super Mumford form $\mu$ on the moduli space of super Riemann surfaces with $\nr$ Ramond Punctures $\msp$. This can then be used to create a measure on $\msp$ whose integral computes scattering amplitudes of superstring theory. The following arguments were heavily inspired by the work done in \cite{belman}, \cite{vormum} and \cite{RSVphy}.

Throughout this section we let $\pi: X \to S$ denote a family of genus $g \geq 2$ SUSY curves with $\nr$ Ramond punctures with $R^i\pi_*\w^j$ locally free and $\nr > 6g-6$. We denote the Ramond divisor by $\mathcal{F}$.

\bigskip

\subsection{Some Riemann-Roch Calculations}~ \label{somerierochcalc}

In the special case that $S$ is a point the structure sheaf $\stsh_X$ admits the global decomposition $\stsh_X = \stsh_{X_{\tred}} \oplus \J$. This allows one to decompose any super holomorphic line bundle $\Ll$ into the direct sum of two ordinary holomorphic line bundles over $X_{\tred}$ as
\begin{equation} \label{decomp1}
\Ll = \Ll_{\text{red}} \oplus \left( \Ll_{\text{red}} \otimes \J \right ).
\end{equation}
Here $\J \subset \stsh_X$ again denotes the sheaf of ideals generated by the odd elements. The summands above are exactly the even and odd parts of $\Ll$. We are most interested in the case $\Ll = \w^{\otimes j} = \w^j$ with our goal being to identify the ranks of these bundles. By the local freeness assumption (and the cohomology and base change Theorem), to compute these ranks it suffices to assume that $S$ is a point. Thus by the decomposition (\ref{decomp1})
\begin{equation} \label{rank1}
\text{rank} \, R^i\pi_*\Ll = h^i(\Ll_{\text{red}}) \, | \,  h^i(\Ll_{\text{red}} \otimes \J)
\end{equation}
for $\Ll$ of rank $1 | 0 $ (and vice versa for $\Ll$ or rank $0 | 1 $).

In Section \ref{ramondpunctures} we saw that the Berezinian sheaf $\w$ was identified with $\mathcal{D}^{-1}(-\F)$. Furthermore, one can easily see \cite{wit1} that if $S$ is a point, we have that $\mathcal{T}_{X_\tred} = \mathcal{D}^2(\mathcal{F})_\tred$. Thus, the distribution $\mathcal{D}$ had degree $1-g-\nr/2$. Arguing in this fashion, i.e. utilizing the classical Riemann-Roch theorem on $X_{\tred}$, a slightly tweaked Proposition \ref{jspin} and the assumption that $\nr > 6g-6$, allows one to complete the tables below of the various ranks of the sheaves $R^i\pi_*\w^j$. The calculations involved are somewhat tedious and we omit them here.
$$
\begin{array}{c|c|c|c}
j & \text{rank } \pi_*\wre^j & \text{rank } \pi_*(\wre^j \otimes \J) & \text{rank } \pi_*\omega^j \\ \hline
-2 &\nr + 3 - 3g & 3\nr/2 + 2 - 2g & \nr + 3 - 3g \, | \, 3\nr/2 + 2 - 2g \\
-1 &\nr/2 + 2 - 2g & \nr + 1 - g & \nr + 1 - g \,|\, \nr/2 + 2 - 2g   \\
0 & 1 & \nr/2 & 1 \, | \, \nr/2 \\
1 & 0 & g & g \, | \, 0 \\
\end{array}
$$
$$
\begin{array}{c|c|c|c}
j & \text{rank } R^1\pi_*\wre^j & \text{rank } R^1\pi_*(\wre^j \otimes \J) & \text{rank } R^1\pi_*\omega^j \\ \hline
-2 & 0 & 0 & 0 \, | \, 0 \\ 
-1 & 0 & 0 & 0 \, | \, 0   \\ 
0 & g & 0 & g \, | \, 0 \\ 
1 & \nr/2 & 1 & 1 \, | \, \nr/2 \\ 
\end{array}
$$

\medskip

Let us now address the issue of local freeness of the sheaves $R^i\pi_*\omega^j$ on the base $S$. Specifically for our purposes in the Ramond case we are interested in the local freeness of those sheaves with $i=0, 1$ and $j=-2, -1, 0, 1$. Common cohomology and base change arguments give immediately that the sheaves $R^i\pi_*\omega^j$ with $i=0,1$, $j=-1, -2$ are indeed locally free. Unfortunately for the others, it seems that there is no elementary argument to guarantee their local freeness. In fact, similar issues have been discussed in the literature before. A result in \cite{bruzzo2} essentially shows that there is no super version of Grauert's classical theorem of algebraic geometry, which would yield the desired result. The author would like to thank E. Witten for his helpful and stimulating comments regarding this issue.

\bigskip

\subsection{The Super Mumford Isomorphism $\lambda_{3/2}\lambda_{1/2}^{-5} \cong \stsh_S$ Explicitly }~

Our goal now is to prove explicitly the super Mumford isomorphism $\lambda_{3/2}\lambda_{1/2}^{-5} \cong \mathcal{O}_S$, following every step carefully to identify the trivializing section $\mu$ corresponding to the image of $1_S$ under the above isomorphism. We will express $\mu$ in terms of bases chosen for the locally free sheaves $R^i\pi_*\omega^j$.

The isomorphism will follow from 3 short exact sequences of sheaves on $X$

\begin{equation}
\label{firstSES} 0 \longrightarrow \Pi\omega \overset{t}{\longrightarrow} \oo \longrightarrow \oo |_{T} \longrightarrow 0,
\end{equation}

\begin{equation}
\label{secondSES} 0 \longrightarrow \oo \overset{t}{\longrightarrow} \Pi\w^{-1} \longrightarrow (\Pi\w^{-1}) |_{T} \longrightarrow 0,
\end{equation}

\begin{equation}
\label{thirdSES} 0 \longrightarrow \Pi\omega^{-1} \overset{t}{\longrightarrow} \w^{-2} \longrightarrow \w^{-2} |_{T} \longrightarrow 0,
\end{equation}
where $t = \Pi t'$ for $t'$ is an odd global section of $\w^{-1}$, and $T$ is the divisor $\{ t = 0 \}$. From these one concludes utilizing Serre duality (noting $B(\Pi \mathcal{F}) = B^{-1}(\mathcal{F})$)

\begin{equation} \label{3eqns}
    B(\oo|_T)\cong \lambda^2_{1/2}, \hspace{.25in} B((\Pi \w^{-1})|_T) \cong \lambda^{-1}_1 \lambda^{-1}_{1/2}, \hspace{.25in} \text{and} \hspace{.25in} B(\w^{-2}|_T) \cong \lambda_{3/2} \lambda_1.
\end{equation}

An important lemma is shown in \cite{vormum} and \cite{vids}, a proof of which is also given in Lemma \ref{canisodiv}, stating that given any invertible sheaves $\mathcal{L}$ and $\mathcal{K}$ of on $X$ and any relative divisor $D$ of dimension $0|1$ over the base, we have a canonical isomorphism $B(\mathcal{L}|_D) \cong B(\mathcal{K}|_D)$. Using this result we get that the left hand sides of the three equations in (\ref{3eqns}) are all canonically identified and thus,
$$ 
\lambda^2_{1/2} \cong \lambda^{-1}_1 \lambda^{-1}_{1/2}, \hspace{.1in} \lambda_{3/2} \lambda_1 \cong \lambda^2_{1/2} \,\,\, \Longrightarrow \,\,\, \lambda_{3/2} \cong \lambda^5_{1/2}, 
$$
giving the super Mumford isomorphism.

We will follow the above argument in detail to identify $\mu$ in terms of specified bases for the sheaves $R^i\pi_*\w^j$.

\bigskip

\subsection{Various Bases}~ \label{variousbases}

To simplify notation we let $r := \nr/2 - g + 1$. We choose a distinguished odd global section $t'$ of $\w^{-1}$ such that $t'_{\tred}$ vanishes to first order at points $q_1, \dots, q_r$. Set $\Pi t' = t$ and near each point $q_k$ choose local superconformal coordinates $ z_k | \tth_k $ centered at $q_k$ so that $t$ expands in these coordinates as
$$ t' \sim z_kf_k(z_k | \tth_k)\bpk. $$
We denote by $T$ the divisor $\{ t' = 0 \} = \{t = 0 \}$ and assume it is disjoint from the Ramond divisor $\mathcal{F}$ (which is an open condition).

\bigskip 

\subsection*{Local Basis for $\pi_*\w$:}~

The rank of $\pi_*\w$ is $g | 0$, thus we choose an (even) basis 
$$ 
\mathcal{B}_{\w} := \{ \varphi_1, \dots, \varphi_g \}.
$$
Near each $q_k$ we expand
$$ \varphi_j \sim (\varphi_j^{k,-} + \varphi_j^{k,+}\tth_k + O(z_k)) [dz_k \, | \, d\tth_k], $$
where $\varphi_j^{k,\pm}$ are even/odd functions from the base. 

\subsection*{Local Basis for $\pi_*\stsh$:}~

We take 
$$ 
\mathcal{B}_{\stsh} := \{1 \, | \, t'\varphi_1, \dots, t'\varphi_g, \xi_1, \dots \xi_{r-1} \},
$$
to be a basis in $\pi_*\oo$ where
$$ \xi_j \sim (\xi_j^{k,-} + \xi_j^{k,+}\tth_k + O(z_k)) $$
near $q_k$.

\subsection*{Local Basis for $\pi_*\w^{-1}$:}~

Let 
$$
\mathcal{B}_{\w^{-1}} := \{ t'^2\varphi_1, \dots, t'^2\varphi_g, t'\xi_1, \dots, t'\xi_{r-1}, \sigma_1, \dots \sigma_r \, | \, t', \tau_1 \dots, \tau_{r-g} \},
$$
be a basis for $\pi_*\w^{-1}$. In local coordinates near each $q_k$ expand
$$ \sigma_j \sim (\sigma_j^{k,-} + \sigma_j^{k,+}\tth_k + O(z_k)) \bpk, $$
$$ \tau_j \sim (\tau_j^{k,+} + \tau_j^{k,-}\tth_k + O(z_k)) \bpk. $$

\subsection*{Local Basis for $\pi_*\w^{-2}$:}~

Let 
\begin{align*}
\mathcal{B}_{\w^{-2}} : = \{t'^2, t'\tau_1&, \dots, t'\tau_{r-g}, \eta_1, \dots, \eta_r \, |
\\
& \, t'^3\varphi_1, \dots, t'^3\varphi_g, t'^2\xi_1, \dots, t'^2\xi_{r-1}, t'\sigma_1, \dots t'\sigma_r, \psi_1, \dots, \psi_r \}, 
\end{align*} 
be a basis for $\pi_*\w^{-2}$ and as above expand near each $q_k$
$$ \eta_j \sim (\eta_j^{k,+} + \eta_j^{k,-}\tth_k + O(z_k))\bpk^2, $$
$$ \psi_j \sim (\psi_j^{k,-} + \psi_j^{k,+}\tth_k + O(z_k))\bpk^2. $$

\subsection*{Local Basis for $\pi_*(\w^j|_T)$:}~

We have singled out a specific odd global section $t'$ of $\w^{-1}$, for which we defined a divisor $T = \{t' = 0\}$. We assume this is disjoint from the Ramond Divisor $\mathcal{F}$ and then take
$$
\mathcal{B}_{\w^j|_T} := \{ [\partial_{z_1}\,|\,\partial_{\tth_1}]^j, \dots, [\partial_{z_r}\,|\,\partial_{\tth_r}]^j \, | \, \theta_1[\partial_{z_1}\,|\,\partial_{\tth_1}]^j, \dots, \theta_r[\partial_{z_r}\,|\,\partial_{\tth_r}]^j \} $$
to be a basis in $\w^j|_T$ (for $j=0$ we denote by $1_k = \bpk^0$ the element which is the function $1$ near each $q_k$, similarly $1_k\tth_k$ will sometimes be shortened to $\tth_k$). To shorten notation we will sometimes use $\varpi_k$ for $\bpk$.

Now by Serre duality we identify $R^1\pi_*\w^j \cong (\pi_*\w^{1-j})^*$, and hence by taking dual bases we get local bases $\mathcal{B}_{\stsh}^*$ and $\mathcal{B}_{\w}^*$  for $R^1\pi_*\w$ and $R^1\pi_*\stsh$ respectively (recall from Section \ref{somerierochcalc} that both $R^1\pi_*\w^{-1}$ and $R^1\pi_*\w^{-2}$ vanish). 

The five (ordered) bases above give rise to generating elements of their respective Berezinian of cohomology, 
\begin{equation} \label{gens_1}
\begin{split}
d_{1/2} & := \ber \mathcal{B}_{\w} \otimes \ber \mathcal{B}^*_{\stsh} \in \lambda_{1/2}\\
d_0 & := \ber \mathcal{B}_{\stsh} \otimes \ber \mathcal{B}^*_{\w} \in \lambda_{0} \\
d_{-1/2} & := \ber \mathcal{B}_{\w^{-1}} \in \lambda_{-1/2} \\
d_{-1} & := \ber \mathcal{B}_{\w^{-2}} \in \lambda_{-1} \\
\delta_{j/2} & := \ber \mathcal{B}_{\w^j|_T} \in B(\w^j|_T).
\end{split}
\end{equation}

\smallskip

\subsection{The First Short Exact Sequence}~ \label{firstSESsec}

We now study in detail the first short exact sequence shown in \eqref{firstSES}. Considering the induced long exact sequence in cohomology and utilizing Serre duality we obtain (Note: $R^1\pi_*(\stsh|_T)$ will vanish as $T$ has relative dimension $0|1$),
\begin{equation} \label{ses1}
0 \longrightarrow \Pi(\pi_*\w) \overset{t}{\longrightarrow} \pi_*\oo \longrightarrow \pi_*(\oo|_{T}) \longrightarrow \Pi(\pi_*(\oo))^* \overset{t^*}{\longrightarrow} (\pi_*(\w))^* \longrightarrow 0.
\end{equation}

Under the first map the (odd) basis $\Pi\{\varphi_1, \dots \varphi_g\}$ of $\Pi(\pi_*\w)$ maps to $\{t\varphi_1, \dots t\varphi_g\}$ (note the presence of $t$ and not $t'$) which is then completed to the chosen basis $\mathcal{B}_{\stsh} = \{ 1 \, | \, t\varphi_1, \dots, t\varphi_g, \xi_1, \dots, \xi_{r-1} \}$ of $\pi_*\oo$. The restriction map then sends the $t\varphi_j$'s to zero and the remaining basis elements to $\{ 1|_T \, | \, \xi_1|_T, \dots, \xi_{r-1}|_T \}$. In terms of our chosen (ordered) basis $\mathcal{B}_{\stsh|_T}$ for $\pi_*(\oo|_T)$, these are in components
\begin{align} \label{res1}
\begin{split}
1|_T = & \sum_{k=1}^{r} 1_k, \\
\xi_j|_T = & \sum_{k=1}^{r} \xi_j^{k,-}1_k + \sum_{k=1}^{r} \xi_j^{k,+}\tth_k1_k.
\end{split}
\end{align}

We will complete $\{ 1|_T \, | \, \xi_1|_T, \dots, \xi_{r-1}|_T \}$ to a basis by taking certain lifts of elements of $\Pi(\pi_*(\oo))^*$ and compare this with $\mathcal{B}_{\stsh|_T}$. To calculate lifts we must analyze the connecting homomorphism $\delta: \pi_*(\stsh|_{T}) \to R^1\pi_*(\Pi \w)$ and its composition with the Serre dual map $\pi_*(\stsh|_{T}) \overset{\delta}{\to} R^1\pi_*(\Pi \w) \to \Pi(\pi_*(\stsh))^*$.

We compute the map $\delta$ using $\bar{D}$ cohomology as explained in Section \ref{dbarcoho}. Let $\mathcal{E}$ denote the sheaf of smooth super functions on $X$, then given any superholomorphic vector bundle $\mathcal{F}$ we have a natural acyclic resolution as in Proposition \ref{dbarSES}. This gives the exact diagram,
\begin{equation} \label{diag1}
\begin{tikzcd}
 & 0 \arrow[d] & 0 \arrow[d]  & 0 \arrow[d] &  \\
0 \arrow[r] & \Pi \w \arrow[d] \arrow[r, "t"] & \stsh \arrow[d] \arrow[r] & \stsh|_{T} \arrow[d] \arrow[r] & 0 \\
0 \arrow[r] & \Pi\w \otimes \mathcal{E}  \arrow[r, "t"] \arrow[d, "\bar{D}"] & \oo \otimes \mathcal{E} \arrow[d, "\bar{D}"] \arrow[r] & \oo|_T \otimes \mathcal{E} \arrow[d, "\bar{D}"] \arrow[r] & 0 \\
0 \arrow[r] & \Pi\w \otimes \w^{0,1} \arrow[r, "t"] \arrow[d] & \oo \otimes \w^{0,1} \arrow[r] \arrow[d] & \oo|_T \otimes \w^{0,1} \arrow[r] \arrow[d] & 0 \\
& 0 & 0  & 0  &  \\
\end{tikzcd}.
\end{equation}

The connecting homomorphism is computed in this diagram by following the zigzag pattern starting from the upper right position at $\stsh|_{T}$ to $\Pi \w \otimes \w^{0,1}$ at the lower left, then looking at the image in the quotient.

We can describe $\delta$ in the following way: let $s$ be a section of $\pi_*(\stsh|_{T})$ and choose a global real valued smooth bump function $\rho$ on $X$ that is identically equal to $1$ in a neighborhood of each $q_k$. Then the expression $\rho s$ defines a lift of $s$ to a global smooth superfunction on $X$. Differentiating and lifting under the multiplication by $t$ map gives that $\delta(s)$ is the cohomology class
\begin{align}
\delta(s) = \left [ \frac{\bar{D}(\rho s)}{t} \right ] \in R^1\pi_*(\Pi \w).
\end{align}

Let $\langle \cdot, \cdot \rangle$ denote the Serre duality pairing, then $\delta(s)$ becomes the functional
\begin{equation} \label{sd1}
\langle \cdot, \delta(s) \rangle = \frac{1}{2 \pi i}\int_{X/S}(-) \frac{\bar{D}(\rho s)}{t}.
\end{equation}
Let us express (\ref{sd1}) in different terms. If one carries out the computation in coordinates one arrives at the following, analogous to the classical situation,
\begin{align} \label{res3}
\begin{split}
\langle h, \delta(s) \rangle & = \frac{1}{2 \pi i}\int_{X/S}h \frac{\bar{D}(\rho s)}{t} \\
& = \sum_{k=1}^r \text{res}_{q_k} \left ( \frac{hs}{t} \right )
\end{split}
\end{align}
for $h \in \pi_*\stsh$. We remark that the meaning of taking the residue at $q_k$ of the expression $hs/t$ means to take the local function defining $s$ near $q_k$ and compute the resulting residue. This gives an alternative description of the functional $\langle \cdot, \delta(s) \rangle = \sum_k \text{res}_{q_k} (-)s/t$.

Now, in view of the exact sequence (\ref{ses1}) we see that the image of the map $\pi_*(\stsh|_{T}) \to \Pi ( \pi_*(\stsh))^*$ sending $s \mapsto \langle \cdot, \delta(s) \rangle$ is the kernel of $t^*: \Pi(\pi_*(\oo))^* \to(\pi_*(\w))^*$ which is $\text{span}\, \Pi \{1^* \, | \, \xi_1^*, \dots, \xi_{r-1}^*\}$. Thus $\langle \cdot, \delta(s) \rangle$ expands as
\begin{equation} \label{exp1}
\langle \cdot, \delta(s) \rangle = \langle 1, \delta(s) \rangle 1^* + \sum_{k=1}^{r-1}\langle \xi_k, \delta(s) \rangle \xi_k^*.
\end{equation}

On the other hand, the kernel of $t^*$ is spanned by the set
$$
\{\langle \cdot, \delta(1_1) \rangle , \dots, \langle \cdot, \delta(1_r) \rangle \, | \, \langle \cdot, \delta(\tth_1) \rangle , \dots, \langle \cdot, \delta(\tth_r) \rangle \},
$$
and hence equation (\ref{exp1}) applied to each member of the (ordered) set above yield the various expressions
\begin{align*}
\langle 1, \delta(1_k) \rangle = \text{res}_{q_k} \left ( \frac{1}{t} \right ), \hspace{.5cm} \langle 1, \delta(\tth_k) \rangle = \text{res}_{q_k} \left ( \frac{\tth_k}{t} \right ),
\end{align*}
\begin{align*}
\langle \xi_j, \delta(1_k) \rangle = \text{res}_{q_k} \left ( \frac{\xi_j}{t} \right ), \hspace{.5cm} \langle \xi_j, \delta(\tth_k) \rangle = \text{res}_{q_k} \left ( \frac{\xi_j \tth_k}{t} \right ),
\end{align*}
which in view of the local expansions of Section \ref{variousbases} and the computation (\ref{res2}), can be encoded in the $(2r \times r)$ matrix
$$ 
A' = 
\begin{bmatrix}
  \text{res}_{q_1} 1/t & \text{res}_{q_1} \xi_1/t  & \dots & \text{res}_{q_1} \xi_{r-1}/t \\
  \vdots & \vdots & &  \vdots \\
  \text{res}_{q_r} 1/t & \text{res}_{q_r} \xi_1/t & \dots & \text{res}_{q_r} \xi_{r-1}/t \\
  \text{res}_{q_1} \tth_1/t & \text{res}_{q_1} \xi_1\tth_1/t  & \dots & \text{res}_{q_1} \xi_{r-1}\tth_1/t \\
  \vdots & \vdots & &  \vdots \\
  \text{res}_{q_r} \tth_r/t & \text{res}_{q_r} \xi_1\tth_r/t & \dots & \text{res}_{q_r} \xi_{r-1}\tth_r/t \\
\end{bmatrix}.
$$

Letting $A = (a_{ij})$ denote any $(r \times 2r)$ left inverse of $A'$ gives us lifts to $\pi_*(\stsh|_{T})$ of the elements $\{ 1^* \, | \, \xi_1^*, \dots, \xi_{r-1}^* \}$,
\begin{equation} \label{lift1}
\begin{split}
\widetilde{1^*} = & \sum_{k=1}^r a_{1,j}1_k + \sum_{k=1}^r a_{1,j+r}\tth_k1_k, \\
\widetilde{\xi_{j-1}^*} = & \sum_{k=1}^r a_{j,k}1_k + \sum_{k=1}^r a_{j,k+r}\tth_k1_k.
\end{split}
\end{equation}
Therefore combining (\ref{res1}) and (\ref{lift1}) we have that the bases 
$$
\{1|_T, \widetilde{\xi^*_1},\dots,\widetilde{\xi^*_{r-1}} \, | \, \xi_1|_T, \dots, \xi_{r-1}|_T, \widetilde{1^*} \}
$$
and
$$ \{ 1_1,\dots, 1_r \, | \, \tth_1, \dots, \tth_r \} $$
of $\pi_*\stsh|_T$ are related by the matrix 
\begin{equation} \label{m0_r}
M_0=
\begin{bmatrix}
    1       & a_{2,1} & \dots &  a_{r,1} & \xi_1^{1,-} & \dots & \xi_{r-1}^{1,-} & a_{1,1}\\
    1       & a_{2,2} & \dots & a_{r,2} & \xi_1^{2,-} &  \dots & \xi_{r-1}^{2,-} & a_{1,2} \\
    1       & a_{2,3} & \dots & a_{r,3} & \xi_1^{3,-} &  \dots & \xi_{r-1}^{3,-} & a_{1,3} \\
    \vdots & \vdots && \vdots & \vdots & & \vdots & \vdots \\
    1       & a_{2,r} & \dots & a_{r,r} & \xi_1^{r,-} &  \dots & \xi_{r-1}^{r,-} & a_{1,r} \\
    0 & a_{2,r+1} & \dots & a_{r,r+1} & \xi_1^{1,+} &  \dots & \xi_{r-1}^{1,+} & a_{1,r+1} \\
    0 & a_{2,r+2} & \dots & a_{r,r+2} & \xi_1^{2,+} &  \dots & \xi_{r-1}^{2,+} & a_{1,r+2} \\
    \vdots & \vdots & \dots & \vdots & \vdots & \dots & \vdots & \vdots \\
	0 & a_{2,2r-1} & \dots & a_{r,2r-1} & \xi_{1}^{r-1,+} &  \dots & \xi_{r-1}^{r-1,+} & a_{1,2r-1} \\
    0 & a_{2,2r} & \dots &  a_{r,2r} & \xi_1^{r,+} & \dots & \xi_{r-1}^{r,+} & a_{1,2r} &
\end{bmatrix}.
\end{equation}
That is, we get
$$ \ber \{ 1|_T, \tilde{\xi_k^*} \, | \, \xi_k|_T, \tilde{1^*} \} = \ber M_0 \, \delta_{0} $$
in $B(\stsh|_T) = \ber \pi_*(\stsh|_T)$. Therefore under the canonical isomorphism
$$ B(\oo|_T) \cong \lambda_{1/2} \lambda_{0} $$
we have the identification
\begin{equation}
\ber \{ 1|_T, \tilde{\xi_k^*} \, | \, \xi_k|_T, \tilde{1^*} \} = \ber M_0 \, \delta_{0} = d_{1/2}d_0. 
\end{equation}

\bigskip

\subsection{The Second Short Exact Sequence}~ \label{secondSESsec}

We move on to analyze the second short exact sequence \eqref{secondSES}. Our specified bases for the sheaves listed here are compatible with this short exact sequence, except for the third term. The induced long exact sequence after Serre duality reads

$$ 0 \longrightarrow \pi_*\oo \overset{t}{\longrightarrow} \pi_*(\Pi \w^{-1}) \longrightarrow \pi_*(\Pi \w^{-1}|_{T}) \longrightarrow (\pi_*\w)^* \longrightarrow 0. $$

We aim to replicate the argument given in Section \ref{firstSESsec} to relate the two bases we have for $\pi_*(\Pi \w^{-1}|_{T})$. The key again is understanding the connecting homomorphism $\delta: \pi_*(\Pi \w^{-1}|_T) \to R^1\pi_*\stsh$ and its composition with the Serre duality map $\pi_*(\Pi \w^{-1}|_T) \overset{\delta}{\to} R^1\pi_*\stsh \to (\pi_*\w)^*$. Here we get an exact diagram analogous to (\ref{diag1}),

\begin{equation} \label{diag2}
\begin{tikzcd}
 & 0 \arrow[d] & 0 \arrow[d]  & 0 \arrow[d] &  \\
0 \arrow[r] & \stsh \arrow[d] \arrow[r, "t"] & \Pi \w^{-1} \arrow[d] \arrow[r] & \Pi \w^{-1}|_{T} \arrow[d] \arrow[r] & 0 \\
0 \arrow[r] & \stsh \otimes \mathcal{E}  \arrow[r, "t"] \arrow[d, "\bar{D}"] & \Pi \w^{-1} \otimes \mathcal{E} \arrow[d, "\bar{D}"] \arrow[r] & \Pi \w^{-1}|_{T} \otimes \mathcal{E} \arrow[d, "\bar{D}"] \arrow[r] & 0 \\
0 \arrow[r] & \stsh \otimes \w^{0,1} \arrow[r, "t"] \arrow[d] & \Pi \w^{-1} \otimes \w^{0,1} \arrow[r] \arrow[d] & \Pi \w^{-1}|_{T} \otimes \w^{0,1} \arrow[r] \arrow[d] & 0 \\
& 0 & 0  & 0  &  \\
\end{tikzcd}.
\end{equation}

Following the same argument as given in Section \ref{firstSESsec}, we can describe the connecting homomorphism $\delta$ as the map
\begin{align}
\begin{split}
\delta:  \pi_*(\Pi &\w^{-1}|_T) \to R^1\pi_*\oo \\
 & s \mapsto \left [ \frac{\bar{D}(\rho s)}{t} \right ]
 \end{split}
\end{align}
where again $\rho$ is a real valued global smooth bump function on $X$ that is identically equal to one in a neighborhood of each $q_k$. This map composed with the Serre dual is identical to (\ref{res3}),

\begin{equation} \label{sd2}
\langle \cdot, \delta(s) \rangle = \frac{1}{2 \pi i}\int_{X/S}(-) \frac{\bar{D}(\rho s)}{t} = \sum_{k=1}^r \text{res}_{q_k} \left( (-) \frac{s}{t} \right ) .
\end{equation}

In terms of the basis $\mathcal{B}_{\w}^*$, each $\langle \cdot, \delta(s) \rangle $ expands as
\begin{equation} \label{exp2}
\langle \cdot, \delta(s) \rangle = \sum_{j=1}^g \langle \varphi_j, \delta(s) \rangle \varphi^*_j.
\end{equation}
Thus, we apply (\ref{exp2}) to each member of the basis $\mathcal{B}_{\w^{-1}|_T}$ and encode it in a $(2r \times g)$ matrix $B'$. To simplify notation let $\varpi_k = \bpk $,

$$ 
B'  = 
\begin{bmatrix}
    \text{res}_{q_1} \varphi_1 \varpi_1/t & \dots & \text{res}_{q_1} \varphi_g \varpi_1/t \\
    \vdots & & \vdots \\
    \text{res}_{q_r} \varphi_1 \varpi_r/t & \dots & \text{res}_{q_r} \varphi_g \varpi_r/t \\
    \text{res}_{q_1} \varphi_1 \tth_1\varpi_1/t & \dots & \text{res}_{q_1} \varphi_g \tth_1\varpi_1/t \\
    \vdots & & \vdots \\
    \text{res}_{q_r} \varphi_1 \tth_r\varpi_r/t & \dots & \text{res}_{q_r} \varphi_g \tth_r\varpi_r/t \\
\end{bmatrix}.
$$
Hence we can invert (non-uniquely) the relationships (\ref{exp2}) encoded in $B'$ by finding any left inverse to $B'$, call it $B = (b_{ij})$ which yields lifts to $\pi_*(\Pi \w^{-1}|_T)$ of the elements $\{ \varphi_1^*, \dots, \varphi_g^* \}$
\begin{equation}
\widetilde{\varphi_k^*} = \sum_{j=1}^r b_{k,j}\varpi_j + \sum_{j=1}^r b_{k,r+j}\tth_j \varpi_j.
\end{equation}

Now in $\pi_*(\Pi \w^{-1}|_T)$ the two bases $\Pi \mathcal{B}_{\w^{-1}|_T}$ and 
$$
\{ \tau_1|_T, \dots, \tau_{r-g}|_T, \widetilde{\varphi_1^*}, \dots, \widetilde{\varphi_g^*} \, | \, \sigma_1|_T, \dots, \sigma_r|_T \},
$$
are related by the matrix $M_{-1/2}$,
\begin{equation} \label{m-1/2_r}
M_{-1/2} = 
\begin{bmatrix}
    \tau_1^{1,-} & \dots & \tau_1^{r,-} & \tau_1^{1,+} & \dots & \tau_1^{r,+} \\
    \vdots && \vdots &\vdots && \vdots \\
    \tau_r^{1,-} & \dots & \tau_r^{r,-} & \tau_r^{1,+} & \dots & \tau_r^{r,+} \\
    b_{1,r+1} & \dots & b_{1,2r} & b_{1,1} & \dots & b_{1,r} \\
    \vdots && \vdots &\vdots && \vdots \\
    b_{g,r+1} & \dots & b_{g,2r} & b_{g,1} & \dots & b_{g,r} \\
    \sigma_1^{1,+} & \dots & \sigma_1^{r,+} & \sigma_1^{1,-} & \dots & \sigma_1^{r,-} \\
    \vdots && \vdots &\vdots && \vdots \\
    \sigma_r^{1,+} & \dots & \sigma_r^{r,+} & \sigma_r^{1,-} & \dots & \sigma_r^{r,-} \\
\end{bmatrix}.
\end{equation}
Therefore the relationship
$$ \ber \{ \tau_1|_T, \dots, \tau_{r-g}|_T, \widetilde{\varphi_1^*}, \dots, \widetilde{\varphi_g^*} \, | \, \sigma_1|_T, \dots, \sigma_r|_T \} = \ber M_{-1/2} \, \delta_{-1/2} $$
along with the canonical isomorphism $\lambda_0 \otimes B(\Pi \w^{-1}|_T) \cong \lambda^{-1}_{-1/2} $, gives 
\begin{equation}
d^{-1}_0 d^{-1}_{-1/2} = \ber M_{-1/2} \delta_{-1/2}.
\end{equation}

\bigskip

\subsection{The Third Short Exact Sequence}~ \label{thirdSESsec}

We analyze the final short exact sequence \eqref{thirdSES}. In this case as $R^1\pi_*(\w^{-1}) = R^1\pi_*(\w^{-2}) = 0$, the induced long exact sequence is actually the short exact sequence

$$ 0 \longrightarrow \pi_*(\Pi \w^{-1}) \overset{t}{\longrightarrow} \pi_*\w^{-2} \longrightarrow \pi_*(\w^{-2}|_{T}) \longrightarrow  0. $$
This allows us to quickly identify a basis of $\pi_*(\w^{-2}|_{T})$ coming from the chosen basis $\mathcal{B}_{\w^{-2}}$, namely $\{ \eta_1|_T , \dots, \eta_r|_T \, | \, \psi_1|_T, \dots, \psi_r|_T \}$. This is related to the basis $\{\varpi_k^2 \, | \, \tth_k\varpi_k^2 \}$ by the matrix
\begin{equation} \label{m-1_r}
M_{-1} = 
\begin{bmatrix}
    \eta_1^{1,+} & \dots & \eta_1^{r,+} & \eta_1^{1,-}  & \dots &  \eta_1^{r,-} \\
    \vdots && \vdots & \vdots && \vdots \\
    \eta_r^{1,+} & \dots & \eta_r^{r,+} & \eta_r^{1,-}  & \dots &  \eta_r^{r,-} \\
    \psi_1^{1,-} & \dots & \psi_1^{r,-} & \psi_1^{1,+}  & \dots &  \psi_1^{r,+} \\
    \vdots && \vdots & \vdots && \vdots \\
    \psi_r^{1,-} & \dots & \psi_r^{r,-} & \psi_r^{1,+}  & \dots &  \psi_r^{r,+} \\
\end{bmatrix}.
\end{equation}
Therefore in $B(\w^{-2}|_T)$ we have
$$ \ber \{ \eta_1|_T , \dots, \eta_r|_T \, | \, \psi_1|_T, \dots, \psi_r|_T \} = \ber M_{-1} \, \delta_{-1}$$
and under the identification $\lambda_{-1/2}^{-1} \otimes B(\w^{-2}|_T) \cong \lambda_{-1}$, we get
\begin{equation}
d_{-1}d_{-1/2} = \ber M_{-1} \delta_{-1}. 
\end{equation}

\bigskip

\subsection{An Expression for $\mu$}~

The calculations done in Sections \ref{firstSESsec}, \ref{secondSESsec} and \ref{thirdSESsec} gave
\begin{equation*}
    \begin{split}
        d_{1/2}d_0 & = \ber M_0 \, \delta_{0}, \\
        d^{-1}_0 d^{-1}_{-1/2} & = \ber M_{-1/2} \, \delta_{-1/2}, \\
        d_{-1}d_{-1/2} & = \ber M_{-1} \, \delta_{-1}.
    \end{split}
\end{equation*} 

Serre duality yields $d_0 = d_{1/2}$, and by the argument in \cite{vormum} and \cite{vids} one identifies for each $j$,  $\delta_{j/2} = \delta_0$. These facts give by elementary algebra

$$
d_{-1} = \frac{\ber M_{-1}  \,\ber M_{-1/2}}{(\ber M_0)^2} \, d_{1/2}^{5}.
$$
Thus, we obtain an explicit expression for the trivializing section $\mu$ as follows:
\begin{thm} \label{mainthm1}
Suppose $\pi: X \to S$ is a family of super Riemann surfaces of genus $g \geq 2$ with $\nr$ Ramond punctures such that:
\begin{enumerate}
    \item $\nr > 6g-6$.
    \item The sheaves $R^i\pi_*\w^j$ are locally free for $i=0,1$, $j=-2, -1, 0, 1$.
\end{enumerate}
Then the super Mumford form $\mu$ may be expressed via the sections chosen in (\ref{gens_1}) as
$$ \mu = d_{-1}d_{1/2}^{-5} \frac{(\ber M_0)^2}{\ber M_{-1}  \,\ber M_{-1/2}} \in \lambda_{-1} \lambda_{1/2}^{-5} \cong \oo_S, $$
where $M_0, M_{-1/2}$ and $M_{-1}$ are given by (\ref{m0_r}), (\ref{m-1/2_r}) and (\ref{m-1_r}) respectively.
\end{thm}

\bigskip

\subsection{A Measure on $\msp$}~ \label{measure_mspr}

Here we follow an idea of E. Witten in \cite{witmeasure}. Thus far we have an explicit formula for the super Mumford form $\mu$, trivializing the line bundle $\lambda_{-1} \lambda_{1/2}^{-5}$ on the moduli space $\msp$. The significance of such a section is that the line bundle $\lambda_{-1}\lambda_{1/2}^{-5}$ is related to the Berezinian of $\msp$.

As discussed in Section \ref{ramondpunctures}, the tangent sheaf to $\msp$ is $R^1\pi_*\mathcal{W}$ where $\mathcal{W}$ is the sheaf of infinitesimal automorphisms, which is seen to be isomorphic (as sheaves of $\mathbb{C}$-vector spaces) to $\mathcal{D}^2$. Hence, by Serre duality and the isomorphism $\w^{-2}(-2\mathcal{F}) = \mathcal{D}^2$ one sees that
\begin{equation} \nonumber
\ber \msp = \ber \Omega^1_{\msp} = \ber \pi_*(\w^3(2\mathcal{F})).
\end{equation}
Noting that $R^1\pi_*\w^3(2\mathcal{F}) = 0$, we will write this as
\begin{equation}
\ber \msp = B(\w^3(2\mathcal{F})).
\end{equation}

Trivially we have the short exact sequence
$$
0 \longrightarrow \w^3 \longrightarrow \w^3(2\mathcal{F}) \longrightarrow \w^3(2\mathcal{F})/\w^3 \longrightarrow 0.
$$
By Corollary \ref{trivL} of the appendix we have canonically the identification
\begin{equation}
B(\w^3(2\mathcal{F})) = B(\w^3) \cong \lambda_{3/2}.
\end{equation}

Now Serre duality identifies $\lambda_{3/2}$ with $\lambda_{-1}$, and thus the super Mumford form $\mu$ can in fact be thought of as a section of $\ber \msp$ valued in a certain line bundle
$$
\mu \in \ber \msp \otimes \lambda_{1/2}^{-5} \cong \stsh_{\msp}.
$$

In bosonic string theory (without punctures), the analogous argument would yield a form similar to $\mu$ such that its modulus squared could genuinely be regarded (in the sense that one had a natural pairing between the analogous factor $\lambda_{1/2}^{-5}$ and its conjugate) as a section of the smooth Berezinian (or simply the determinant in this case) of the moduli space of Riemann surfaces. The celebrated result of Belavin and Knizhnik \cite{belkniz} states that this procedure indeed yields the integrand of the bosonic string partition function, the so-called Polyakov measure.

In superstring theory the super Mumford form $\mu$ plays a similar role to its bosonic counterpart, in that it can be paired with something analogous to its complex conjugate to yield a genuine measure. However, the story is a bit more complicated. The interested reader can learn more in E. Witten's notes \cite{wit4}.

\section{The Neveu-Schwarz Puncture Case}

Suppose now we have a family $\pi: X \to S$ of SUSY curves of genus $g \geq 2$ with $n_{NS}$ Neveu-Schwarz punctures. We will reproduce the arguments of \cite{vormum} and \cite{RSVphy} to write down the explicit formula for the associated super Mumford form $\mu$. We then discuss how this form can be used to create a genuine measure on $\mathfrak{M}_{g;n_{NS}}$.

As in the Ramond case, we make local freeness assumptions on the higher direct images $R^i\pi_*\w^j$ and describe $\mu$ in terms of chosen local bases for these sheaves. Here specifically we work with $R^i\pi_*\w^j$ for $i=0,1, \, j=0, 1, 2, 3$. We then relate this to a section of $\ber \mathfrak{M}_{g;n_{NS}}$. As the first part of the following argument is identical to the one found in \cite{vormum} and \cite{RSVphy}, we quickly review the procedure to establish notation but omit some details.

We also assume that we are working over the component of $\mathfrak{M}_{g;n_{NS}}$ corresponding to an odd spin structure. That is, for the relative Berezinian sheaf $\w$ we assume that $\pi_*\w$ has rank $g | 1$. Thus on each fiber the reduction $\Pi \w_{\tred}$ gives an odd nondegenerate theta characteristic.

We choose an odd global section $\nu' \in \w$ and consider the short exact sequence

$$ 0 \longrightarrow \stsh_X \overset{\nu}{\longrightarrow} \Pi\w \longrightarrow (\Pi\w)|_{D} \longrightarrow 0 $$
where $\nu = \Pi \nu'$ and $D = \{ \nu = 0 \} = \{ \nu' = 0\}$. This short exact sequence and the two others obtained by twisting by $\w$ and $\w^2$ is what we focus on. Similar to the argument given in the Ramond case, using these three short exact sequences we can produce the super Mumford isomorphism $\lambda_{3/2} \cong \lambda_{1/2}^5$.

As $\nu'$ is an odd global section of $\w$ its divisor $D$ has the property that its reduction $D_{\tred}$ is a finite sum of $g-1$ points (which we assume to be distinct),
$$
D_{\tred} = \sum_{j=1}^{g-1} p_j.
$$
For each $j$ we choose local superconformal coordinates $z_j \, | \, \zeta_j$ centered at $p_j$.

\bigskip

\subsection{Bases}~

We choose specific local bases here of the sheaves $R^i\pi_*\w^j$ and $\pi_*\w|_D$ and analyze their compatibility to the above mentioned short exact sequences. The ranks of these various sheaves are easily computable using the same techniques as used in Section \ref{somerierochcalc} of the Ramond case.

\bigskip

\subsection*{Basis in $\pi_*\stsh_X$:}~

The rank of $\pi_*\stsh_X$ is $1 | 1$ and thus we take the local basis 
$$\mathcal{B}_{\stsh_X} := \{ 1 \, | \, \xi \},$$
where $\xi$ expands near each $p_k$ as
$$
\xi \sim ( \xi^{k, -} + \xi^{k, +}\zeta_k + O(z_j) )
$$
where $\xi^{k,\pm}$ are some even/odd functions from the base $S$.

\subsection*{Basis in $\pi_*\w$:}~

The rank of $\pi_*\w$ is $g  | 1$ and we take the local basis 
$$
\mathcal{B}_{\w} := \{ \varphi_1, \cdots, \varphi_{g-1}, \nu' \xi \, | \, \nu' \},
$$
expanding each $\varphi_j$ near $p_k$ as
$$
\varphi_j \sim ( \varphi_j^{k, +} + \varphi_j^{k, -}\zeta_k + O(z_j) ) [dz_j \, | \, d\zeta_j].
$$

\subsection*{Basis in $\pi_*\w^2$:}~

The rank of $\pi_*\w^2$ is $g  | 2g-2$. We take the local basis 
$$
\mathcal{B}_{\w^2} := \{ \nu'^2, \chi_1, \dots, \chi_{g-1} \, | \, \nu' \varphi_1, \cdots, \nu'\varphi_{g-1}, \nu'^2 \xi, \psi_1, \cdots, \psi_{g-2} \},
$$
expanding each $\chi_j$ and $\psi_j$ near $p_k$ as
$$
\chi_j \sim ( \chi_j^{k, +} + \chi_j^{k, -}\zeta_k + O(z_j) ) [dz_j \, | \, d\zeta_j]^2,
$$
$$
\psi_j \sim ( \psi_j^{k, -} + \psi_j^{k, +}\zeta_k + O(z_j) ) [dz_j \, | \, d\zeta_j]^2.
$$

\subsection*{Basis in $\pi_*\w^3$:}~

The rank of $\pi_*\w^3$ is $3g-3  | 2g-2 $. We take the local basis 
\begin{align} \nonumber
\mathcal{B}_{\w^3} := \{ \nu'^2 \varphi_1, \cdots, \nu'^2\varphi_{g-1}, \nu'^3\xi, &\nu'\psi_1, \cdots, \nu'\psi_{g-1}, \sigma_1, \cdots, \sigma_{g-1} \, \\ & | \, \nu'^3, \nu' \chi_1, \cdots, \nu' \chi_{g-1}, \rho_1, \cdots, \rho_{g-2}  \}, \nonumber
\end{align}
expanding each $\sigma_j$ and $\rho_j$ near $p_k$ as
$$
\sigma_j \sim ( \sigma_j^{k, +} + \sigma_j^{k, -}\zeta_k + O(z_j) ) [dz_j \, | \, d\zeta_j]^3,
$$
$$
\rho_j \sim ( \rho_j^{k, -} + \rho_j^{k, +}\zeta_k + O(z_j) ) [dz_j \, | \, d\zeta_j]^3.
$$

\subsection*{Basis in $\pi_*(\w^{-1})$:}~

The rank of $\pi_*(\w^{-1})$ is $1  |  0$ and we take the local basis 
$$
\mathcal{B}_{\w^{-1}} := \{ \xi / \nu'\}.
$$

\subsection*{Basis in $\pi_*\w^j|_{D}$:}~

With the aid of the specific chosen local coordinates $z_j \, | \, \zeta_j$ we take the local basis (for $j \geq 0$)
$$
\mathcal{B}_{\w^j|_D} := \{ [dz_1|d\zeta_1]^j, \cdots, [dz_{g-1}|d\zeta_{g-1}]^j \, | \, \zeta_1[dz_1|d\zeta_1]^j, \cdots, \zeta_{g-1}[dz_{g-1}|d\zeta_{g-1}]^j \}.
$$

Finally we utilize Serre duality, the canonical isomorphisms $R^1\pi_*\w^j \cong (\pi_*\w^{1-j})^*$, to construct local bases for $R^1\pi_*\stsh_X$, $R^1\pi_*\w$, $R^1\pi_*\w^2$, and $R^1\pi_*\w^3$ by taking the image of the corresponding dual bases of those already specified. We denote these bases by $\mathcal{B}_{\w^j}^*$.

We set
$$
\lambda_{j/2} := B(\w^j).
$$
Using these bases, we consider the following local generators of the various $\lambda_{j/2}$,
\begin{equation} \label{gens_2}
\begin{split}
d_0 & := \ber \mathcal{B}_{\stsh_X} \otimes \ber \mathcal{B}^*_{\w} \in \lambda_0 \\
d_{1/2} & : = d_0 \in \lambda_{1/2} \\
d_1 & := \ber \mathcal{B}_{\w^2} \otimes \ber \mathcal{B}^*_{\w^{-1}} \in \lambda_1 \\
d_{3/2} & := \ber \mathcal{B}_{\w^3} \in \lambda_{3/2} \\
\delta_{j/2} & := \ber \mathcal{B}_{\w^j|_D} \in B(\w^j|_D). \\
\end{split}
\end{equation}

\bigskip

\subsection{Relating the Chosen Bases}~ \label{relate_bases_ns}

The first short exact sequence is
$$ 0 \longrightarrow \stsh_X \overset{\nu}{\longrightarrow} \Pi\w \longrightarrow (\Pi\w)|_{D} \longrightarrow 0 $$
which gives the long exact sequence on cohomology (after using Serre duality)
$$
0 \longrightarrow \pi_*\stsh_X \overset{\nu}{\longrightarrow} \Pi\pi_*\w \longrightarrow \Pi\pi_*(\w)|_{D} \longrightarrow (\pi_*\w)^* \longrightarrow \Pi (\pi_*\stsh_X)^* \longrightarrow 0. 
$$

Following the work of \cite{vormum} or \cite{RSVphy} we conclude under the canonical isomorphism $\lambda_{1/2}^{-1} \cong \lambda_0 \otimes B(\w|_{D})^{-1}$ that
\begin{equation} \label{id1}
d_{1/2}^{-1} = \ber M_1 \, d_0 \delta_{1/2}^{-1}
\end{equation}
where $M_1$ is the block matrix
\begin{equation} \label{m1_ns}
M_1 = 
\begin{pmatrix}
   A_1 \\
   B^t_1
\end{pmatrix}
\end{equation}
where $A_1$ is the $(g-1) \times (2g-2)$ matrix
$$
A_1 = 
\left (
\begin{array}{ccc | ccc}
\varphi_1^{1,+} & \dots & \varphi_1^{g-1,+} & \varphi_1^{1,-} & \dots & \varphi_1^{g-1,-} \\
\vdots & & \vdots & \vdots & & \vdots \\
\varphi_{g-1}^{1,+} & \dots & \varphi_{g-1}^{g-1,+} & \varphi_{g-1}^{1,-} & \dots & \varphi_{g-1}^{g-1,-}
\end{array}
\right )
$$
and $B_1$ is any left inverse of $A_1$.

Similarly, the exact sequence
$$
0 \longrightarrow \Pi\w \overset{\nu}{\longrightarrow} \w^2 \longrightarrow (\w^2)|_{D} \longrightarrow 0,
$$
yields
$$
0 \longrightarrow \Pi \pi_*\w \overset{\nu}{\longrightarrow} \pi_*\w^2 \longrightarrow \pi_*(\w^2)|_{D} \longrightarrow \Pi (\pi_*\stsh_X)^* \longrightarrow (\pi_*\w^{-1})^* \longrightarrow 0.
$$
Arguing again as in \cite{vormum} and \cite{RSVphy} we conclude
\begin{equation} \label{id2}
    d_1 = \ber M_2 \, d_{1/2}^{-1} \delta_{1},
\end{equation}
where $M_2$ is the $(2g-2) \times (2g-2)$ square matrix

\begin{equation} \label{m2_ns}
M_2 = 
\left (
\begin{array}{ccc | ccc}
\chi_1^{1,+} & \dots & \chi_1^{g-1,+} & \chi_1^{1,-} & \dots & \chi_1^{g-1,-} \\
\vdots & & \vdots & \vdots & & \vdots \\
\chi_{g-1}^{1,+} & \dots & \chi_{g-1}^{g-1,+} & \chi_{g-1}^{1,-} & \dots & \chi_{g-1}^{g-1,-} \\
\phantom{asdf} \\
\hline \\
\psi_{1}^{1,-} & \dots & \psi_{1}^{g-1,-} & \psi_{1}^{1,+} & \dots & \psi_{1}^{g-1,+} \\
\vdots & & \vdots & \vdots & & \vdots \\
\psi_{g-2}^{1,-} & \dots & \psi_{g-2}^{g-1,-} & \psi_{g-2}^{1,+} & \dots & \psi_{g-2}^{g-1,+} \\
0 & \dots & 0 & 0 & \dots & 1
\end{array}
\right ).
\end{equation}
\medskip

Lastly, the final short exact sequence
$$
0 \longrightarrow \w^2 \overset{\nu}{\longrightarrow} \Pi\w^3 \longrightarrow (\Pi\w^3)|_{D} \longrightarrow 0
$$
gives
$$
0 \longrightarrow \pi_*\w^2 \overset{\nu}{\longrightarrow} \Pi\pi_*\w^3 \longrightarrow \Pi\pi_*(\w^3)|_{D} \longrightarrow (\pi_*\w^{-1})^* \longrightarrow 0.
$$
Thus, under $\lambda_{3/2}^{-1} \cong \lambda_1 B(\w^3|_D)^{-1}$ we conclude
\begin{equation} \label{id3}
    d_{3/2}^{-1} = \ber M_3 \, d_1 \delta_{3/2}^{-1},
\end{equation}
where $M_3$ is the $(2g-2) \times (2g-2)$ square matrix

\begin{equation} \label{m3_ns}
M_3 = 
\left (
\begin{array}{ccc | ccc}
\rho_1^{1,+} & \dots & \rho_1^{g-1,+} & \rho_1^{1,-} & \dots & \rho_1^{g-1,-} \\
\vdots & & \vdots & \vdots & & \vdots \\
\rho_{g-2}^{1,+} & \dots & \rho_{g-2}^{g-1,+} & \rho_{g-2}^{1,-} & \dots & \rho_{g-2}^{g-1,-} \\
\xi_1^{-1} & \dots & 0 & 0 & \dots & 0 \\
\hline \\
\sigma_{1}^{1,-} & \dots & \sigma_{1}^{g-1,-} & \sigma_{1}^{1,+} & \dots & \sigma_{1}^{g-1,+} \\
\vdots & & \vdots & \vdots & & \vdots \\
\sigma_{g-1}^{1,-} & \dots & \sigma_{g-1}^{g-1,-} & \sigma_{g-1}^{1,+} & \dots & \sigma_{g-1}^{g-1,+}
\end{array}
\right ).
\end{equation}

\medskip 

Under the canonical isomorphism guaranteed by Lemma \ref{canisodiv} we have that 
$$
\delta_{j/2}^{(-1)^{j-1}}=\delta_{1/2}.
$$
Combining this with the identifications (\ref{id1}), (\ref{id2}) and (\ref{id3}), we get the desired formula for the super Mumford form $\mu$ as in \cite{vormum} and \cite{RSVphy}.
\begin{thm} \label{mainthm2} \emph{(}\cite{vormum}, \cite{RSVphy}\emph{)}
Suppose $\pi:X \to S$ is a family of super Riemann surfaces of genus $g \geq 2$ such that:
\begin{enumerate}
    \item The sheaves $R^i\pi_*\w^j$ are locally free for $i=0,1$, $j=0,1,2, 3$.
    \item $\pi_*\w$ has rank $g|1$.
\end{enumerate}
Then the super Mumford form $\mu$ may be expressed via the sections chosen in (\ref{gens_2}) as
$$
\mu = d_{3/2} d_{1/2}^{-5} \frac{\ber M_3 \, \ber M_2}{(\ber M_1)^2},
$$
where $M_1, M_2$ and $M_3$ are given by (\ref{m1_ns}), (\ref{m2_ns}) and (\ref{m3_ns}) respectively.
\end{thm}

\bigskip

\subsection{Relation to $\ber \mathfrak{M}_{g;n_{NS}}$}~

In the situation without any punctures, such a formula for $\mu$ is of immediate interest as the Berezinian of $\mathfrak{M}_g$ \emph{is} simply $\lambda_{3/2}$, hence $\mu$ is interpreted as a section of the Berezinian of supermoduli space valued in a certain line bundle. However, when one considers punctures the story is a bit different, as it is no longer true that $\lambda_{3/2}$ is $\ber \mathfrak{M}_{g;n_{NS}}$. In Section \ref{nspuctures_sec} we saw that $\mathcal{T}_{\mathfrak{M}_{g;n_{NS}}} \cong R^1\pi_*\mathcal{W} \cong R^1\pi_*\mathcal{D}^2(-N)$ for $N = \sum_{k=1}^{n_{NS}} \text{div}(s_k)$ the Neveu-Schwarz divisor. Serre duality then gives that instead we have $\ber \mathfrak{M}_{g;n_{NS}} \cong B(\w^3(N)) = \ber \pi_*\w^3(N)$.

Let $N_{\tred} = \sum q_k$ be the reduction. Then each $q_k$ is a divisor in $X_{\tred}$ that is a single point in each fiber of $\pi$. We choose an even global section of $\w^3(N)$, call it $\tau$, that vanishes to exactly first order on each $q_k$. For each $k$ we choose local coordinates $x_k \, | \, \theta_k$ such that $\tau$ near each $q_k$ is
$$
\tau \sim x_k(a_k + b_k\tth_k + O(x_k)) [dx_k \, | \, d\tth_k].
$$

$\tau$ then induces a short exact sequence
$$
0 \longrightarrow \w^3 \overset{\tau}{\longrightarrow}\w^3(N) \longrightarrow \w^3(N)|_N \longrightarrow 0,
$$
which in fact gives the short exact sequence on cohomology
$$
0 \longrightarrow \pi_*\w^3 \overset{\tau}{\longrightarrow} \pi_*\w^3(N) \longrightarrow \pi_*\w^3(N)|_N \longrightarrow 0.
$$

The rank of $\pi_*\w^3(N)$ is $3g-3 + n_{NS} \, | \, 2g-2 + n_{NS}$, thus we construct a local basis for $\pi_*\w^3(N)$ in the following way. First we consider the image of $\mathcal{B}_{\w^3}$ under $\tau$ and complete it to a basis. Namely we construct
$$
\mathcal{B}_{\w^3(N)} := \tau \mathcal{B}_{\w^3} \cup \mathcal{B}'
$$
where $\mathcal{B}'$ is
$$
\mathcal{B}' = \{ \alpha_1, \cdots, \alpha_{n_{NS}} \, | \, \beta_1, \cdots, \beta_{n_{NS}} \}.
$$
We expand each $\alpha_j$ and $\beta_j$ near $q_k$ as
$$
\alpha_j \sim ( \alpha_j^{k, +} + \alpha_j^{k, -}\tth_k + O(x_j) ) [dx_j \, | \, d\tth_j]^3,
$$
$$
\beta_j \sim ( \beta_j^{k, -} + \beta_j^{k, +}\tth_k + O(x_j) ) [dx_j \, | \, d\tth_j]^3,
$$
and let
\begin{align} \nonumber
\mathcal{B}_{\w^3(N)|_N} := & \\ \{ [dx_1|d\tth_1]^3, \cdots, [dx_{n_{NS}}|d\tth_{n_{NS}}]^3 \, | \, &\tth_1[dx_1|d\tth_1]^3, \cdots, \tth_{n_{NS}}[dx_{n_{NS}}|d\tth_{n_{NS}}]^3 \}. \nonumber
\end{align}
Putting
\begin{equation} \label{gens_3}
    \begin{split}
        \delta^N_{3/2} & := \ber \mathcal{B}_{\w^3(N)|_N}, \\ 
        d_{3/2}^N & := \ber \mathcal{B}_{\w^3(N)},
    \end{split}
\end{equation}
we easily see that in the canonical identification 
$$ 
B(\w^3(N)) \cong \lambda_{3/2} \, B(\w^3(N)|_N)
$$
we have
$$
d_{3/2}^N = \ber M' \, d_{3/2} \delta^N_{3/2},
$$
where

\begin{equation} \label{m'}
M' = 
\left (
\begin{array}{ccc | ccc}
\alpha_1^{1,+} & \dots & \alpha_1^{n_{NS},+} & \alpha_1^{1,-} & \dots & \alpha_1^{n_{NS},-} \\
\vdots & & \vdots & \vdots & & \vdots \\
\alpha_{n_{NS}}^{1,+} & \dots & \alpha_{n_{NS}}^{n_{NS},+} & \alpha_{n_{NS}}^{1,-} & \dots & \alpha_{n_{NS}}^{n_{NS},-} \\
\phantom{asdf} \\
\hline \\
\beta_{1}^{1,-} & \dots & \beta_{1}^{n_{NS},-} & \beta_{1}^{1,+} & \dots & \beta_{1}^{n_{NS},+} \\
\vdots & & \vdots & \vdots & & \vdots \\
\beta_{n_{NS}}^{1,-} & \dots & \beta_{n_{NS}}^{n_{NS},-} & \beta_{n_{NS}}^{1,+} & \dots & \beta_{n_{NS}}^{n_{NS},+}
\end{array}
\right ).
\end{equation}
Combining this discussion with that of the Section \ref{relate_bases_ns} we obtain the following corollary.
\begin{cor} \label{maincor} Suppose $\pi:X \to S$ is a family of super Riemann surfaces of genus $g \geq 2$ with $n_{NS}$ Neveu-Schwarz punctures such that:
\begin{enumerate}
    \item The sheaves $R^i\pi_*\w^j$ are locally free for $i=0,1$, $j=0,1,2, 3$.
    \item $\pi_*\w$ has rank $g|1$.
\end{enumerate}
Then via the sections defined in (\ref{gens_2}), (\ref{gens_3}) and matrices in (\ref{m1_ns}), (\ref{m2_ns}), (\ref{m3_ns}), (\ref{m'}), the expression
$$
\mu^N := d^N_{3/2} (\delta^N_{3/2})^{-1} d_{1/2}^{-5} \frac{\ber M_3 \, \ber M_2}{(\ber M_1)^2 \, \ber M'} 
$$
gives a trivializing section of the line bundle
$$
 \ber \mathfrak{M}_{g;n_{NS}} \otimes B(\w^3(N)|_N)^{-1} \otimes \lambda_{1/2}^{-5}
$$
on $\mathfrak{M}_{g; n_{NS}}$.
\end{cor} 

Thus the constructed object $\mu^N$ can be viewed as a section of the Berezinian of the moduli space $\mathfrak{M}_{g:n_{NS}}$ with values in a particular line bundle.

The utility of such a formula for $\mu^N$ is that it indeed can be used to construct a measure on $\mathfrak{M}_{g; n_{NS}}$. The process of constructing this measure depends on the particular type of superstring theory one is working in, heterotic or Type II for example. In \cite{witmeasure} such a process is described, however it assumes the object one starts with is a section of $\ber \mathfrak{M}_{g;n_{NS}} \otimes \lambda_{1/2}^{-5}$ rather than what is given in Corollary \ref{maincor}, in a section of $\ber \mathfrak{M}_{g;n_{NS}} \otimes B(\w^3(N)|_N)^{-1} \otimes \lambda_{1/2}^{-5}$. In calculating scattering amplitudes, one inserts so called vertex operators at each puncture. The collection of them can be thought of as sections of $B(\w^3(N)|_N)$. Hence, after multiplying with the form $\mu^N$ we indeed arrive at a section of $\ber \mathfrak{M}_{g;n_{NS}} \otimes \lambda_{1/2}^{-5}$. The details of this discussion can be found in \cite{wit4} and \cite{witmeasure}.

\appendix

\section{A Few Technical Results}~

Here we develop a few technical statements used in the main arguments of the paper. These were motivated by \cite{witmeasure}.

\bigskip

Suppose we have a family of SUSY curves with $n_R$ Ramond punctures $\pi: X \to S$. Denote by $\w = \ber X/S$, the relative Berezinian sheaf. Let $\mathcal{F}$ denote the Ramond divisor and decompose it $\mathcal{F} = \sum_{k}^{n_R} \mathcal{F}_k$ into its $n_R$ minimal components. Recall (say, in the complex topology) that near a Ramond divisor $\mathcal{F}_k$ there are coordinates $x \, | \, \tth$ such that the divisor $\mathcal{F}_k$ is given by $\{ x = 0 \}$ and the distribution $\mathcal{D}$ is generated by
$$
D^*_\tth := \frac{\partial}{\partial \tth} + x \tth \frac{\partial}{\partial x}.
$$
Here the distinguished subbundle of $\Omega^1_{X/S}$ is $\mathcal{D}^{-2}(-\mathcal{F})$ and admits a generator
$$
\varpi^*_\tth := dx - x\tth d\tth.
$$

Coordinates near a Ramond puncture for which $D^*_\tth$ (or $\varpi^*_\tth$) generate $\mathcal{D}$ (resp. $\mathcal{D}^{-2}(-\mathcal{F})$) are called \emph{superconformal}. A superconformal change of coordinates near a Ramond puncture is a change of coordinates $z \, | \, \zeta$ such that one still has $\mathcal{F}_k = \{z = 0 \}$ and that $D^*_{\zeta}$ is a $\stsh_X$-multiple of $D^*_{\tth}$. One can also phrase this condition equivalently as the form $\varpi^*_{\zeta}$ is a multiple of $\varpi_{\tth}^*$.

It turns out that the possible choices of superconformal coordinates near a Ramond puncture is restricted. In fact, this is heavily exploited by E. Witten in \cite{wit3} to define the notion of odd periods of closed holomorphic one-forms on such a family of Ramond punctured SUSY curves. Witten phrases this constraint on coordinates as ``The odd coordinate $\tth$ is defined up to sign and a shift by an odd constant."
\begin{lem} \label{superconformallemma}
Let $x \, | \, \tth$ denote superconformal coordinates near a Ramond puncture $\mathcal{F}_k$. Any superconformal change of coordinates $z \, | \, \zeta$ can be expressed as
$$
z = f(x) + \lambda(x)\tth
$$
$$
\zeta = \psi(x) + g(x)\tth
$$
for even $f,g$ and odd $\psi, \lambda$. We then have 
\begin{enumerate}
    \item \label{g_square_1} $g(0)^2 = 1$, and
    \item \label{lambda_prime_psi_0} $\lambda'(0)\psi(0) = 0$.
\end{enumerate}
\end{lem}
\begin{proof}
After some tedious calculations one finds that the condition $\varpi^*_{\zeta}$ is proportional to $\varpi^*_{\tth}$ is
$$
-\left( \frac{\partial z}{\partial x} - z \zeta \frac{\partial \zeta}{\partial x} \right ) x \tth = \left ( \frac{\partial z}{\partial \tth} - z \zeta \frac{\partial \zeta}{\partial \tth} \right ).
$$
In terms of the functions $f,g, \lambda$ and $\psi$ this condition is the pair of conditions
$$
\lambda(x) - f(x)g(x)\psi(x) = 0,
$$
and
$$
f(x)g(x)^2 + \lambda(x)\psi(x)g(x) = xf'(x) - xf(x)\psi(x) \psi'(x).
$$
The first of these two conditions says that $\lambda$ and $\psi$ are proportional, hence their product vanishes. Using this and dividing by $f(x)$ in the second equation gives (note $f \neq 0$ away from $x = 0$)
\begin{equation} \label{app1}
    g(x)^2 = \frac{x}{f(x)} f'(x) - x \psi(x) \psi'(x).
\end{equation}
As the change of coordinates was superconformal, the divisor $\mathcal{F}_k$ was given as both the zero locus of $x$ and $z$, hence in particular $f(0) = 0$. This implies that the ratio $x/f(x) \to 1/f'(0)$ as $x \to 0$. Hence, taking $x \to 0$ in (\ref{app1}) yields $g(0)^2 = 1$, giving (\ref{g_square_1}). Assertion (\ref{lambda_prime_psi_0}) follows immediately from $\lambda(x) = f(x)g(x)\psi(x)$, recalling that $f(0)=0$ and $\psi(x)^2=0$. 
\end{proof}

Lemma \ref{superconformallemma}, will allow us to trivialize $\ber \pi_*(\oo/\oo(-2\mathcal{F}))$ canonically over the base $S$. Combined with Lemma \ref{canisodiv}, this will give a natural trivialization of $\ber \pi_*(\w^3(2\mathcal{F})/\w^3)$. This result proved significant in Section \ref{measure_mspr} as it allowed us to connect the super Mumford form constructed in Theorem \ref{mainthm1} with sections of $\ber \msp$.

\begin{lem} \label{ber_trivial_2f}
The Berezinian of the locally free $\stsh_S$-module $\pi_*(\oo/\oo(-2\mathcal{F}))$ is canonically trivial,
$$ \ber \pi_*(\oo/\oo(-2\mathcal{F})) \cong \stsh_S. $$
\end{lem}
\begin{proof}
We work locally on $S$. First, decompose $\mathcal{F} = \sum_k^{n_R} \mathcal{F}_k$ into its connected components. Then $\ber \pi_*(\oo/\oo(-2\mathcal{F})) = \otimes_k \ber \pi_*(\oo/\oo(-2\mathcal{F}_k))$ and so it suffices to show the result for each $\mathcal{F}_k$. To simplify notation, for the remainder of the proof write $\mathcal{F}$ for some $\mathcal{F}_k$.

Choose superconformal coordinates $x \, | \, \tth$ near $\mathcal{F} = \{ x = 0 \}$. With these coordinates one can trivialize $\ber \pi_*(\oo/\oo(-2\mathcal{F}))$ by the element
\begin{equation} \label{nat_elt_ber}
    \sigma_{x|\tth} = [1, x \, | \, \tth, x \tth],
\end{equation}
where $1, x, \tth, x\tth$ in (\ref{nat_elt_ber}) are to be understood as their images in $\oo/\oo(-2\mathcal{F})$. We claim that the element $\sigma_{x|\tth}$ is in fact canonical, in the sense that if $z \, | \, \zeta$ is another choice of superconformal coordinates, we then have $\sigma_{x|\tth} = \sigma_{z|\zeta}$. Indeed, for such a change of coordinates, write as in Lemma \ref{superconformallemma}
$$
z = f(x) + \lambda(x)\tth,
$$
$$
\zeta = \psi(x) + g(x)\tth.
$$
Looking at their images in the quotient $\oo/\oo(-2\mathcal{F})$, we see that modulo $\oo(-2\mathcal{F})$
\begin{equation*}
    \begin{split}
        z & = f'(0)x + \lambda'(0)x\tth, \\
        \zeta & = \psi(0) + \psi'(0)x + g(0)\tth + g'(0)x\tth, \\
        z \zeta & = f'(0)\psi(0)x + f'(0)g(0) x\tth.
    \end{split}
\end{equation*}

Hence, in $\pi_*(\oo/\oo(-2\mathcal{F}))$, the change of basis matrix $A$ from $\{1, x \, | \, \tth, x\tth \}$ to $\{ 1, z, \, | \, \zeta, z \zeta \}$ is given by
$$
A = 
\begin{pmatrix}
1 & 0 & \psi(0) & 0 \\
0 & f'(0) & \psi'(0) & f'(0)\psi(0) \\
0 & 0 & g(0) & 0 \\
0 & \lambda'(0) & g'(0) & f'(0)g(0) 
\end{pmatrix}.
$$
Recalling from Lemma \ref{superconformallemma} that $g(0)^2=1$ and $\lambda'(0)\psi(0) = 0$, a quick calculation will show that $\ber A = 1$. Thus the element $\sigma = \sigma_{x|\tth} = \sigma_{z|\zeta}$ is independent of the choice of superconformal coordinates. 

This local argument glues to a global canonical isomorphism 
$$
\ber \pi_*(\oo/\oo(-2\mathcal{F})) \cong \stsh_S.
$$
\end{proof}

Now, with the aid of Lemma \ref{canisodiv} we obtain

\begin{cor} \label{trivL} There is a canonical isomorphism
$$
\left( \ber \pi_*(\oo|_{\mathcal{F}}) \right)^{\otimes 2} \cong \oo_S.
$$
Hence, in particular for $\w$ the relative Berezinian sheaf, we get a natural identification
$$
\ber \pi_*(\w^3(2\mathcal{F})/\w^3) \cong \oo_S.
$$
\end{cor}
\begin{proof}
By Lemma \ref{ber_trivial_2f}, $\ber \pi_*(\oo/\oo(-2\mathcal{F})) \cong \stsh_S$ is naturally trivial. On the other hand, by Lemma \ref{canisodiv}
\begin{equation*}
    \begin{split}
        \ber \pi_*(\oo/\oo(-2\mathcal{F})) & \cong \ber(\oo|_{\mathcal{F}}) \otimes \ber(\oo(-\mathcal{F})|_{\mathcal{F}}) \\ 
        & \cong \left( \ber \pi_*(\oo|_{\mathcal{F}}) \right)^{\otimes 2}.
    \end{split}
\end{equation*}
From here, it follows that
\begin{equation*}
    \begin{split}
        \ber \pi_*(\w^3(2\mathcal{F})/\w^3) & \cong \ber(\w^3(2\mathcal{F})|_{\mathcal{F}}) \otimes \ber(\w^3(\mathcal{F})|_{\mathcal{F}}) \\ 
        & \cong \left( \ber \pi_*(\oo|_{\mathcal{F}}) \right)^{\otimes (-2)} \\
        & \cong \oo_S.
    \end{split}
\end{equation*}

\end{proof}

\section{A Lemma for the Super Mumford Isomorphism}

\begin{lem} \label{canisodiv}
Let $\pi : D \to S$ be a smooth proper morphism of complex supervarieties (or complex supermanifolds) of relative dimension $0 | 1$. Let $\Ll$ and $\Kk$ be two invertible sheaves on $D$ of rank $1|0$. We then have a canonical isomorphism 
$$
\ber \pi_* \Ll \cong \ber \pi_*\Kk.
$$
\end{lem}
\begin{proof}
For a sufficiently small open set $U$ in $S$, setting $V= \pi^{-1}(U)$ we can find an isomorphism $\varphi: \Ll|_V \to \Kk|_V$. This is possible as the map $\pi$ is of relative dimension $0|1$. The isomorphism $\varphi$ then induces an isomorphism $\ber \pi_*\varphi: \ber \pi_* \Ll|_V \to \ber \pi_* \Kk|_V$. If $\psi: \Ll|_V \to \Kk|_V$ was a possibly different isomorphism, it would differ from $\varphi$ by an $\stsh_D|_V$ automorphism. Every such automorphism is multiplication by an even invertible function $f \in \stsh_D|_V$. We denote this by $m_f$ so that $ \psi = m_f \circ \varphi$. Hence $\psi$ would then induce the isomorphism 
$$ \ber \pi_*m_f \circ \ber \pi_* \varphi : \ber \pi_* \Ll|_V \to \ber \pi_* \Kk|_V. $$

The key fact is that, by shrinking $U$ if necessary, we have that $\ber \pi_*m_f = 1$. Indeed for small affine $U = \text{Spec}(A)$, the space $D$ above is of the form $\text{Spec}(A[\alpha])$ for some odd parameter $\alpha$ (or more precisely a disjoint union of such spaces in the \'etale topology). Then the $\stsh_D|_U$ automorphism $m_f$ is equivalent to the $A[\alpha]$ automorphism (also denoted) $m_f$ which is multiplication by an even invertible element $f \in A[\alpha]$. The ring $A[\alpha]$ is a free $A$-module of rank $1|1$ and thus $\ber \pi_*m_f$ is the Berezinian of $m_f$ viewing $m_f$ as a map of $A$-modules. Writing $f = f_0 + f_1\alpha$ in components we easily see that with respect to the basis $\{1, \alpha \}$, the matrix of $m_f$ is
$$
m_f = 
\begin{pmatrix}
    f_0 & 0 \\
    f_1 & f_0
\end{pmatrix}.
$$
Thus, $\ber m_f$ = 1 as claimed (note that $f_0 \neq 0$ by assumption).

Thus we see that the invertible sheaves on $S$, $\ber \pi_*\Ll$ and $\ber \pi_* \Kk$ are canonically isomorphic for all sufficiently small open sets $U$ in $S$. Hence, we get a natural global isomorphism.
\end{proof}

\section*{Acknowledgements}

I'd like to thank the organizers of the 2015 Supermoduli workshop at the Simons Center for Geometry and Physics and those who put in the effort to put those excellent lectures online. I am grateful to the speakers of the workshop R. Donagi, P. Deligne, E. Witten, E. D'Hoker and D. H. Phong whose efforts illuminated many interesting concepts. I'd like to thank E. Witten for his insight and several valuable comments, and especially A. Voronov for his guidance and the frequent helpful discussions. The author would also like to thank the anonymous referee whose comments and suggestions greatly improved the quality of this work.

\bibliographystyle{plain}
\bibliography{bibfile}

\begin{thebibliography}{10}

\bibitem{belman}
A.~A. Beilinson and Yu.~I. Manin.
\newblock The {M}umford {F}orm and the {P}olyakov {M}easure in {S}tring
  {T}heory.
\newblock {\em Comm. Math. Phys.}, 107(3):359--376, 1986.

\bibitem{belkniz}
A.~A. Belavin and V.~G. Knizhnik.
\newblock Complex {G}eometry and the {T}heory of {Q}uantum {S}trings.
\newblock {\em Zh. \`Eksper. Teoret. Fiz.}, 91(2):364--390, 1986.

\bibitem{bruzzo2}
U.~Bruzzo and J.~A. Dom\'inguez~P\'erez.
\newblock Line {B}undles {O}ver {F}amilies of ({S}uper) {R}iemann {S}urfaces.
  {II}. {T}he {G}raded {C}ase.
\newblock {\em J. Geom. Phys.}, 10(3):269--286, 1993.

\bibitem{vids}
P.~Deligne, E.~Witten, R.~Donagi, E.~D'Hoker, and D.H. Phong.
\newblock Supermoduli {W}orkshop {V}ideo {L}ectures.
\newblock \url{http://scgp.stonybrook.edu/archives/10356}, May 2015.

\bibitem{DPhong}
Eric D'Hoker and D.~H. Phong.
\newblock Lectures on {T}wo-{L}oop {S}uperstrings.
\newblock In {\em Superstring theory}, volume~1 of {\em Adv. Lect. Math.
  (ALM)}, pages 85--123. Int. Press, Somerville, MA, 2008.

\bibitem{donwit}
R.~Donagi and E.~Witten.
\newblock Supermoduli {S}pace is not {P}rojected.
\newblock In {\em String-{M}ath 2012}, volume~90 of {\em Proc. Sympos. Pure
  Math.}, pages 19--71. Amer. Math. Soc., Providence, RI, 2015.

\bibitem{gid2}
S.~B. Giddings and P.~Nelson.
\newblock Line {B}undles on {S}uper {R}iemann {S}urfaces.
\newblock {\em Comm. Math. Phys.}, 118(2):289--302, 1988.

\bibitem{smoothfromcomplex}
C.~Haske and R.~O. Wells, Jr.
\newblock Serre {D}uality on {C}omplex {S}upermanifolds.
\newblock {\em Duke Math. J.}, 54(2):493--500, 1987.

\bibitem{modSRS}
C.~LeBrun and M.~Rothstein.
\newblock Moduli of {S}uper {R}iemann {S}urfaces.
\newblock {\em Comm. Math. Phys.}, 117(1):159--176, 1988.

\bibitem{man1}
Yu.~I. Manin.
\newblock {\em Gauge {F}ield {T}heory and {C}omplex {G}eometry}, volume 289 of
  {\em Grundlehren der Mathematischen Wissenschaften [Fundamental Principles of
  Mathematical Sciences]}.
\newblock Springer-Verlag, Berlin, 1988.
\newblock Translated from the Russian by N. Koblitz and J. R. King.

\bibitem{RSV}
A.~A. Rosly, A.~S. Schwarz, and A.~A. Voronov.
\newblock Geometry of {S}uperconformal {M}anifolds.
\newblock {\em Comm. Math. Phys.}, 119(1):129--152, 1988.

\bibitem{RSVphy}
A.~A. Rosly, A.~S. Schwarz, and A.~A. Voronov.
\newblock Superconformal {G}eometry and {S}tring {T}heory.
\newblock {\em Comm. Math. Phys.}, 120(3):437--450, 1989.

\bibitem{vormum}
A.~A. Voronov.
\newblock A {F}ormula for the {M}umford {M}easure in {S}uperstring {T}heory.
\newblock {\em Funktsional. Anal. i Prilozhen.}, 22(2):67--68, 1988.

\bibitem{elmsg}
A.~A. Voronov, Yu.~I. Manin, and I.~B. Penkov.
\newblock Elements of {S}upergeometry.
\newblock In {\em Current problems in mathematics. {N}ewest results, {V}ol.\
  32}, Itogi Nauki i Tekhniki, pages 3--25. Akad. Nauk SSSR, Vsesoyuz. Inst.
  Nauchn. i Tekhn. Inform., Moscow, 1988.
\newblock Translated in J. Soviet Math. {{\bf{5}}1} (1990), no. 1, 2069--2083.

\bibitem{wit1}
E.~Witten.
\newblock Notes on {S}uper {R}iemann {S}urfaces and their {M}oduli, 2012.
\newblock arXiv:1209.2459.

\bibitem{wit4}
E.~Witten.
\newblock Superstring {P}erturbation {T}heory {R}evisited, 2012.
\newblock arXiv:1209.5461.

\bibitem{witmeasure}
E.~Witten.
\newblock Notes on {H}olomorphic {S}tring and {S}uperstring {T}heory {M}easures
  of {L}ow {G}enus, 2013.
\newblock arXiv:1306.3621.

\bibitem{wit3}
E.~Witten.
\newblock {The {S}uper {P}eriod {M}atrix With {R}amond {P}unctures}.
\newblock {\em J. Geom. Phys.}, 92:210--239, 2015.

\end{thebibliography}
\nocite{*}

\textsc{Daniel J. Diroff: Department of Mathematics, University of Minnesota, Minneapolis, MN 55455. Email: \emph{dirof003@umn.edu}}

\end{document}